\documentclass[a4paper,onecolumn,11pt]{quantumarticle}
\pdfoutput=1
\usepackage[utf8]{inputenc}
\usepackage[english]{babel}
\usepackage[T1]{fontenc}
\usepackage{amsmath}
\usepackage{hyperref}

\usepackage{tikz}
\usepackage{lipsum}

\usepackage{graphicx}
\usepackage{amsmath,amsfonts,amssymb,amscd}
\usepackage{mathtools}
\mathtoolsset{showonlyrefs}
\usepackage{ascmac}
\usepackage{mathrsfs}
\usepackage{amsthm}
\usepackage{dsfont}
\usepackage{enumitem}
\usepackage{mathtools}
\usepackage{braket}
\usepackage{bm}
\usepackage{verbatim}
\usepackage{framed}
\usepackage{listings}
\usepackage{tikz}
\usepackage{booktabs}
\usepackage{multirow}

\usepackage[dvipsnames]{xcolor}
\usepackage{hyperref}
\hypersetup{
	pdfnewwindow=true,
	colorlinks=true,
	citecolor=Blue,
	linkcolor=Blue,
	urlcolor=Blue
}
\definecolor{shadecolor}{gray}{0.9}

\newcommand{\rmid}{\mathrm{id}}

\DeclareMathOperator{\Tr}{Tr}

\newcommand{\eq}[1]{Eq.~\eqref{#1}}

\newcommand{\minimize}{\mathrm{Minimize}}
\newcommand{\maximize}{\mathrm{Maximize}}
\newcommand{\subto}{\mathrm{subject\:to}}

\renewcommand{\ge}{\geqslant}
\renewcommand{\le}{\leqslant}

\newcommand{\ketbra}[2]{\ket{{#1}}\!\!\bra{{#2}}}

\newcommand{\cC}{\mathcal{C}}
\newcommand{\cD}{\mathcal{D}}

\newcommand{\cH}{\mathcal{H}}
\newcommand{\cI}{\mathcal{I}}
\newcommand{\cJ}{\mathcal{J}}

\newcommand{\cL}{\mathcal{L}}

\newcommand{\cN}{\mathcal{N}}

\newcommand{\cR}{\mathcal{R}}
\newcommand{\cS}{\mathcal{S}}

\newcommand{\cX}{\mathcal{X}}
\newcommand{\cY}{\mathcal{Y}}

\newcommand{\rA}{\mathrm{A}}
\newcommand{\rB}{\mathrm{B}}

\newcommand{\rH}{\mathrm{H}}

\newcommand{\rR}{\mathrm{R}}

\newcommand{\rX}{\mathrm{X}}

\newcommand{\sC}{\mathscr{C}}

\newcommand{\sF}{\mathscr{F}}

\newcommand{\sM}{\mathscr{M}}

\newcommand{\E}{\mathsf{E}}
\newcommand{\F}{\mathsf{F}}
\newcommand{\G}{\mathsf{G}}

\newcommand{\K}{\mathsf{K}}
\newcommand{\M}{\mathsf{M}}
\newcommand{\N}{\mathsf{N}}

\newcommand{\one}{\mathds{1}}
\newcommand{\zero}{\mathds{O}}

\newcommand{\mbC}{\mathbb{C}}
\newcommand{\mbN}{\mathbb{N}}

\renewcommand{\emph}[1]{{\textbf{#1}}}

\newtheorem{theorem}{Theorem}
\newtheorem{proposition}{Proposition}
\newtheorem{corollary}{Corollary}

\theoremstyle{definition}
\newtheorem{definition}{Definition}

\newtheorem{example}{Example}

\begin{document}

\title{Robustness-based bounds for approximate joint realizations of incompatible quantum channels}

	\author{Shintaro Minagawa}
	\affiliation{Aix-Marseille University, CNRS, LIS, Marseille, France}
	\email{minagawa.shintaro@gmail.com}
	
	\author{Ryo Takakura}
	\affiliation{Center for Quantum Information and Quantum Biology, The University of Osaka, 1-2 Machikaneyama, Toyonaka, Osaka 560-0043, Japan\looseness=-1}
	\affiliation{Department of Physics, Graduate School of Science, The University of Osaka, 1-1 Machikaneyama, Toyonaka, Osaka 560-0043, Japan\looseness=-1}
	\email{takakura.ryo.qiqb@osaka-u.ac.jp}
	
	\author{Kensei Torii}
	\affiliation{Department of Mathematical Informatics, Nagoya University, Furo-cho, Chikusa-ku, Nagoya 464-8601, Japan}
	\email{toriikensei@nagoya-u.jp}
\maketitle

\begin{abstract}
    Several quantum channels may fail to admit an exact joint realization, a property known as incompatibility. 
    Channel incompatibility unifies various quantum limitations, such as measurement uncertainty, the no-information-without-disturbance principle, and the no-cloning and no-broadcasting theorems.
    We address how accurately incompatible channels can be approximated by the marginals of a single joint channel, quantified by the minimum total diamond-norm error over all joint realizations.
    We establish two-sided bounds relating it to the generalized robustness of incompatibility (RoI), a noise-based resource quantifier of incompatibility.
    For identity channels, we determine both the RoI and the minimum total error exactly. 
    The upper bound is attained by the optimal symmetric universal cloner, while the prefactor in the lower bound is asymptotically optimal as the dimension tends to infinity.
    Applying this result to measurement channels provides a unified, model-independent framework encompassing error-error and information-error-disturbance tradeoffs.
    Furthermore, our robustness-based evaluation of disturbance outperforms an algebraic bound for every non-trivial POVM in dimensions up to six.
\end{abstract}

\section{Introduction}
	Quantum theory is governed by various no-go theorems, including the impossibility of simultaneous measurements \cite{Busch2016,guhne2023colloquium}, the ``no information without disturbance" principle \cite{busch2009no,dariano2017quantum}, and the no-broadcasting (no-cloning) theorem \cite{Wootters1982,Dieks1982,Wiesner1983,barnum1996noncommuting,barnum2007generalized}. 
	Despite arising in different contexts, these limitations are comprehensively described by one concept: incompatibility of quantum channels \cite{heinosaari2017incompatibility}. 
	In this framework, joint unmeasurability manifests as the incompatibility of measurement channels, the information-disturbance tradeoff emerges from the incompatibility between an informative measurement and the identity channel, and the no-broadcasting theorem reflects the incompatibility of two identity channels \cite{heinosaari2016invitation}.
	
	By definition, incompatible channels cannot be jointly and exactly realized, implying that any approximate joint realization inevitably forces a tradeoff in implementation accuracy.
	While quantum measurement theory extensively formalizes such limitations through various uncertainty relations \cite{heisenberg1927uber,ozawa2003universally,ozawa2004uncertainty,watanabe2011uncertainty,busch2013proof,branciard2013error,Buscemi2014} and information-disturbance tradeoffs \cite{maccone2006information,buscemi2006information,maccone2007entropic,kretschmann2008information,heinosaari2013qualitative,dariano2020information}, the connection between those results and the strength of the incompatibility was not always clear.
	Recognizing incompatibility as a quantifiable quantum resource \cite{Haapasalo2015,Skrzypczyk2019,Carmeli2019,uola2019quantifying,mori2020operational,Uola2020,Ducuara2025} naturally raises a fundamental question: how does the inherent strength of incompatibility quantitatively constrain this tradeoff? 
	A general quantitative relation directly connecting the generalized robustness of incompatibility of channels to the error of approximate joint realization remains to be clarified.
	
	In this article, we fill this gap by means of the generalized robustness of incompatibility (RoI) of quantum channels \cite{Haapasalo2015,mori2020operational,Uola2020}.
    First, we measure the accuracy of an approximate joint realization for a given set of incompatible channels using the diamond norm~\cite{kitaev1997quantum} between each marginal channel of the joint channel and the corresponding given channel. 
    We then provide both upper and lower bounds for the minimum total error, the sum of the diamond norms minimized across all possible joint realizations, by using the RoI.
    We also observe that our upper bound is achieved by optimal symmetric universal cloning~\cite{keyl1999optimal}, and further, that the lower bound becomes asymptotically tight as the dimension approaches infinity.
    
    This shows that RoI quantitatively lower bounds the error tradeoff in any approximate joint realization of incompatible channels.
	Generalized robustness is a standard measure in quantum resource theories \cite{vidal1999robustness,Steiner2003,datta2009max,skrzypczyk2019robustness,oszmaniec2019operational,takagi2019general}, quantifying the minimal noise required to erase a given resource; in this case, rendering incompatible channels compatible. 
	While RoI exhibits favorable properties as a resource monotone \cite{chitambar2019quantum,gour2025quantum} and can be efficiently computed via semidefinite programming (SDP), its operational meaning has been tied to advantage in discrimination tasks rather than the physical difficulty of joint realization \cite{Skrzypczyk2019,Carmeli2019,uola2019quantifying,mori2020operational,Uola2020,Ducuara2025}.
	Our results overcome this limitation, establishing a direct link between operational limits in traditional quantum mechanics and the resource-theoretic strength of incompatibility in modern quantum information. 
    
	Applying this general result to measurement channels yields a novel Heisenberg-type error-error uncertainty relation bounded by the resourcefulness of incompatibility.
	Moreover, we reveal a quantitative relation between the ability of a measurement to extract information and its incompatibility with the identity channel. 
	By combining this relation with our general inequality, we derive a rigorous information-error-disturbance tradeoff.
	The primary contribution of this work is not merely to reproduce known uncertainty relations, but to provide a unified framework where these tradeoffs naturally emerge from the principle of quantum channel incompatibility. 
	We furthermore show that this resource-theoretic limit strictly improves known algebraic bounds \cite{heinosaari2013qualitative} in two- to six-dimensional systems for any non-trivial POVMs.
	
	\section{Preliminaries}
	In this paper, we only consider quantum systems associated with finite-dimensional Hilbert spaces $\cH_\rA, \cH_\rB, \ldots$, and following Refs.~\cite{nielsen2010quantum,wilde2017quantum,watrous2018theory}, we briefly explain the basic notions of quantum theory.
	The set of all linear operators from $\cH_\rA$ to $\cH_\rB$ is written as $\cL(\cH_\rA,\cH_\rB)$, endowed with the trace norm $\|\cdot\|_1$, the operator norm $\|\cdot\|$, and the Hilbert-Schmidt inner product $\langle Y,X\rangle:=\Tr[Y^\dagger X]$.
	We write $\cL(\cH_\rA)\equiv\cL(\cH_\rA,\cH_\rA)$ and denote the subset of all Hermitian operators by $\cL_{\rH}(\cH_\rA)$.
	The zero and identity operators are denoted by
	$\zero_\rA$ and $\one_\rA$, respectively.
	For a linear map $\Psi_{\rA\to\rB}\colon\cL(\cH_\rA)\to\cL(\cH_\rB)$, its adjoint $\Psi^\dagger_{\rA\to\rB}\colon\cL(\cH_\rB)\to\cL(\cH_\rA)$ is defined via $\langle \Psi^\dagger_{\rA\to\rB}(Y_\rB),X_\rA\rangle=\langle Y_\rB,\Psi_{\rA\to\rB}(X_\rA)\rangle$.
	
	Each quantum system $\rA$ is associated with a Hilbert space $\cH_\rA$, and a composite system $\rA\rB$ with the tensor product $\cH_\rA\otimes\cH_\rB$, which is denoted as $\cH_{\rA\rB}$ in the following.
	A general measurement on the system is given by a positive operator-valued measure, in short, a \emph{POVM} $\E_\rA^\cX:=\{E_\rA^{x}\}_{x\in\cX}$ satisfying $E_\rA^{x}\succeq\zero_\rA$ and $\sum_{x\in\cX}E_\rA^{x}=\one_\rA$.
	We denote the set of all POVMs by $\sM(\rA)$.
	We say a POVM is \emph{uninformative} (trivial) if its elements are proportional to the identity, i.e., of the form $\{p_x\one_\rA\}_{x\in\cX}$.
	A completely positive and trace-preserving (CPTP) map $\Phi_{\rA\to\rB}:\cL(\cH_\rA)\to\cL(\cH_\rB)$ is called a \emph{quantum channel} describing an evolution from a system $\rA$ to $\rB$.
	We denote the set of all quantum channels by $\sC(\rA\to\rB)$ and write $\Phi_\rA\equiv \Phi_{\rA\to\rA}$ when the output system $\rB$ is the same as $\rA$.
	The \emph{identity channel} $\rmid_\rA$ is a quantum channel satisfying $\rmid_{\rA}(\rho_\rA)=\rho_\rA~(\forall\rho_\rA\in\cL(\cH_\rA))$.
	A \emph{measurement channel} \cite{Holevo2020} for a POVM $\E^\cX_\rA$ is defined as a quantum-to-classical channel via $\Gamma^{\E^\cX_\rA}(\rho_\rA):=\sum_{x\in\cX}\Tr[E^x_\rA\rho_\rA]\ketbra{x}{x}_\rX$, where $\{\ket{x}\}_{x\in\cX}$ is an orthonormal basis of a reference system $\rX$ (hereafter often omitted).
	A \emph{CP-instrument} is a family of CP and trace non-increasing linear maps $\{\cI_{\rA\to\rB}^{x}\}_{x\in\cX}$ such that $\sum_{x\in\cX}\cI_{\rA\to\rB}^{x}$ is TP.
	It models a quantum measurement process \cite{ozawa1984quantum};
	for the measurement of $\{\cI_{\rA\to\rB}^{x}\}_{x\in\cX}$ on a quantum state $\rho_\rA$, the probability of observing the outcome $x\in\cX$ and the post-measurement state are respectively $\Tr[\cI_{\rA\to\rB}^{x}(\rho_\rA)]$ and $\cI_{\rA\to\rB}^{x}(\rho_\rA)/\Tr[\cI_{\rA\to\rB}^{x}(\rho_\rA)]$.
	
	The \emph{diamond norm distance} (see, e.g., Refs.~\cite{kitaev1997quantum,watrous2009semidefinite,wilde2017quantum, watrous2018theory}) between quantum channels $\Lambda_{\rA\to\rB},\Xi_{\rA\to\rB}\in\sC(\rA\to\rB)$ is defined as
	\begin{equation}
		\begin{split}
			\|\Lambda_{\rA\to\rB}-\Xi_{\rA\to\rB}\|_\diamond:=\sup_{n\in\mathbb N}\sup_{\rho_{\rA n}\in\cS(\cH_\rA\otimes\mathbb C^n)}\quad\|(\Lambda_{\rA\to\rB}\otimes\rmid_n)(\rho_{\rA n})-(\Xi_{\rA\to\rB}\otimes\rmid_n)(\rho_{\rA n})\|_1.
		\end{split}    
	\end{equation}
	It determines the optimal success probability in discriminating $(\Lambda_{\rA\to\rB},\Xi_{\rA\to\rB})$ \cite{wilde2017quantum,watrous2018theory}.
	Refs.~\cite{watrous2009semidefinite,ben-aroya2010on,watrous2013simpler} showed that computing the diamond norm reduces to a semidefinite program (SDP)~\cite{boyd2004convex,watrous2018theory,skrzypczyk2019robustness}.

	\section{Incompatibility of quantum channels}
	Let us start with the definition of channel incompatibility. 
	\begin{definition}
		\label{def:incomp_channel}
		A set $\{\Lambda^{(i)}_{\rA\to\rB_i}\}_{i=1}^n$ of $n\in\mbN$ channels whose input spaces are the same is called \emph{compatible} if there exists a channel (a \textit{joint channel}) $\Theta_{\rA\to\rB_1\cdots\rB_n}\in \sC(\rA\to\rB_1\cdots\rB_n)$ satisfying
		\begin{align*}
			\Tr_{\rB_{j\neq i}}\circ~\Theta_{\rA\to\rB_1\cdots\rB_n}&=\Lambda^{(i)}_{\rA\to\rB_i},
		\end{align*} 
		where $\Tr_{\rB_{j\neq i}}$ is the partial trace over the systems other than $\rB_i$.
		Otherwise, $\{\Lambda^{(i)}_{\rA\to\rB_i}\}_{i=1}^n$ is called \emph{incompatible}.
	\end{definition}
	Incompatibility of quantum channels physically describes that the operations cannot be realized simultaneously.  
	A typical example is the \textit{no-broadcasting theorem} (or \textit{no-cloning theorem} for pure states) \cite{barnum1996noncommuting,barnum2007generalized,Wiesner1983,Dieks1982,Wootters1982}, indicating that there is no joint channel for the pair of identity channels $(\rmid_{\rA}, \rmid_{\rA})$.
	One can also consider incompatibility of a pair of a channel and a POVM, $(\Lambda_{\rA\to\rB},\E_\rA^\cX)$.
	It is compatible if there is a CP-instrument $\{\cI_{\rA\to\rB}^{x}\}_{x\in\cX}$ satisfying $\Lambda_{\rA\to\rB}=\sum_{x\in\cX}\cI_{\rA\to\rB}^{x}$ and $E^{x}_\rA=\cI_{\rA\to\rB}^{x\dagger}(\one_\rB)~(\forall x\in\cX)$.
	A quantum channel $\Lambda_{\rA\to\rB}$ compatible with a POVM $\E^\cX_\rA$ is called an \emph{$\E^\cX_\rA$-channel} \cite{heinosaari2013qualitative}.
	As an example, when $\Lambda_{\rA\to\rB}$ is the identity channel, the only class of POVMs compatible with it is uninformative POVMs, known as \textit{no-information without disturbance principle}~\cite{busch2009no,dariano2017quantum}.
	We can also consider the measurement incompatibility of a pair $(\E_\rA^{\cX},\F_\rA^{\cY})$ of POVMs as the non-existence of a joint POVM $\G_\rA^{\cX\cY}$ that reproduces the initial POVMs via its marginals: $\sum_{y\in \cY}G_\rA^{xy}=E_\rA^{x}~(\forall x\in\cX),\;\sum_{x\in \cX}G_\rA^{xy}=F_\rA^{y}~(\forall y\in\cY).$
	These two types of incompatibility can be seen as specific examples of channel incompatibility, as the following proposition shows \cite{heinosaari2017incompatibility}.
	\begin{proposition}\label{proposition:channel-measurement}
		With a measurement channel $\Gamma^{\E^\cX_\rA}$ for a POVM $\E^\cX_\rA$, a channel-POVM pair $(\Lambda_{\rA\to\rB},\E_\rA^\cX)$ is compatible if and only if the channel pair $(\Lambda_{\rA\to\rB},\Gamma^{\E^\cX_\rA}\!)$ is compatible~\cite[Proposition 7]{heinosaari2017incompatibility}. 
		Similarly, a POVM pair $(\E_\rA^\cX, \F_\rA^\cY)$ is compatible if and only if the channel pair $(\Gamma^{\E^\cX_\rA}, \Gamma^{\F_\rA^\cY})$ is compatible~\cite[Proposition 5]{heinosaari2017incompatibility}.
	\end{proposition}
	
	Next, we consider quantifying incompatibility.
	The \textit{generalized robustness of incompatibility} (\textit{RoI}) is one of the quantifiers particularly analyzed for measurement incompatibility \cite{Heinosaari2015,Uola2015,Haapasalo2015,Designolle2019} and used in convex resource theories \cite{Skrzypczyk2019,uola2019quantifying,Ducuara2025}.
	It can be extended to channels \cite{Uola2020,mori2020operational} and plays a crucial role in this paper.
	\begin{definition}
		\label{def_gen-rob}
		For a given set of channels $\{\Lambda^{(i)}_{\rA\to\rB_i}\}_{i=1}^n$, we define the RoI of this set as 
		\begin{equation}
			\begin{split}
				R(\{\Lambda^{(i)}_{\rA\to\rB_i}\}_{i=1}^n):=&\min\Big\{
				r\ge 0~\big|~\cN_{\rA\to\rB_i}^{(i)}\in\sC(\rA\to\rB_i)~(i=1,\dots,n)\\
				&
				\qquad\mathrm{s.t.}\left\{\!
				\frac{\Lambda^{(i)}_{\rA\to\rB_i}\!+\!r\cN_{\rA\to\rB_i}^{(i)}}{1+r}\!\right\}_{i=1}^n\mbox{:~compatible}
				\Big\}.
			\end{split}
		\end{equation}
		Similarly, the RoI of a channel and POVM pair $(\Lambda_{\rA\to\rB},\E^\cX_\rA)$ is as follows:
		\begin{equation}\label{eq:RoI_channel-POVM}
			\begin{split}
				R(\Lambda_{\rA\to\rB},\E^\cX_\rA)
				:=&\min\Big\{
				r\ge 0~\big|~\cN_{\rA\to\rB}\in\sC(\rA\to\rB),\;\N^\cX_\rA\in\sM(\rA)\\
				&\mathrm{s.t.}\left(\frac{\Lambda_{\rA\to\rB}+r\cN_{\rA\to\rB}}{1+r},\frac{\E_{\rA}^\cX+r\N_{\rA}^{\cX}}{1+r}\right)\mbox{: compatible}\Big\}.
			\end{split}
		\end{equation} 
		The RoI for a POVM pair $(\E^\cX_\rA,\F^\cY_\rA)$ is as follows:
		\begin{equation}\label{eq:robustness_POVMs}
			\begin{split}
				R(\E^\cX_\rA,\F^\cY_\rA)&:=\min\Big\{
				r\ge 0~\big|~\M_{\rA}^{\cX},\N^\cY_\rA\in\sM(\rA)\mathrm{s.t.}\left(\frac{\E_{\rA}^\cX+r\M_{\rA}^{\cX}}{1+r},\frac{\F_{\rA}^\cY+r\N_{\rA}^{\cY}}{1+r}\right)\mbox{: compatible}\Big\}.
			\end{split}
		\end{equation}
	\end{definition}
	
	RoI captures the resilience of incompatibility against noise and thus can be interpreted as quantifying the strength of incompatibility.
	In addition, it is non-increasing under postprocessing and is faithful in the sense that $R(\Lambda_{\rA\to\rB_1},\Xi_{\rA\to\rB_2})=0$ if and only if $(\Lambda_{\rA\to\rB_1},\Xi_{\rA\to\rB_2})$ is compatible.
	In other words, it is a resource monotone in the resource theory of incompatibility, where the free set is the set of compatible pairs and free operations are postprocessing of channels \cite{heinosaari2017incompatibility,Guerini2017}.
	As a quantitative extension of Proposition \ref{proposition:channel-measurement}, the RoI involving POVMs reduces to that of a channel-channel pair.
	With the measurement channel $\Gamma^{\E^\cX_\rA}$ of a POVM $\E^\cX_\rA$, it holds that $R(\Lambda_{\rA\to\rB},\E^\cX_\rA)=R(\Lambda_{\rA\to\rB},\Gamma^{\E^\cX_\rA})$ and similarly $R(\E^\cX_\rA,\F^\cY_\rA)=R(\Gamma^{\E^\cX_\rA},\Gamma^{\F^\cY_\rA})$~\cite{mori2020operational}.
	All these forms of RoI can be computed via SDP (see Appendix \ref{Appendix:SDP for incompatibility}).
	
	\section{Bounding minimum total error by the generalized robustness of incompatibility}
	We define the minimum total error among all joint realizations for a given set of channels.
    \begin{definition}
        Let $\{\Lambda^{(i)}_{\rA\to\rB_i}\}_{i=1}^n$ be a set of quantum channels.
        The minimum total error among all joint realizations is defined as follows:
        \begin{equation}
            E(\{\Lambda^{(i)}_{\rA\to\rB_i}\}_{i=1}^n):=\min_{\Theta_{\rA\to\rB_1\cdots\rB_n}\in\sC(\rA\to\rB_1\cdots\rB_n)}\sum_{i=1}^n\|\Lambda^{(i)}_{\rA\to\rB_i}-\Tr_{\rB_{j\neq i}}\circ~\Theta_{\rA\to\rB_1\cdots\rB_n}\|_\diamond.
        \end{equation}
    \end{definition}
    
    The following theorem provides upper and lower bounds of the minimum total error by using RoI (for the proof, see Appendix~\ref{Appendix:Proof of Theorem1}).
	\begin{theorem}
		\label{thm:UR_GR}
		Let $\{\Lambda^{(i)}_{\rA\to\rB_i}\}_{i=1}^n$ be a set of quantum channels. 
		It holds that
		\begin{equation}\label{eq:RoI-error}
			2R(\{\Lambda^{(i)}_{\rA\to\rB_i}\}_{i=1}^n)\le E(\{\Lambda^{(i)}_{\rA\to\rB_i}\}_{i=1}^n)\le \frac{2nR(\{\Lambda^{(i)}_{\rA\to\rB_i}\}_{i=1}^n)}{1+R(\{\Lambda^{(i)}_{\rA\to\rB_i}\}_{i=1}^n)}.
		\end{equation}
	\end{theorem}
    From this theorem, when we have two channels $\Lambda_{\rA\to\rB_1}$ and $\Xi_{\rA\to\rB_2}$, for any approximate joint realization $\Theta_{\rA\to\rB_1\rB_2}$, it holds that
    \begin{equation}
        2R(\Lambda_{\mathrm A\to\mathrm B_1},\Xi_{\mathrm A\to\mathrm B_2})\le \|\Lambda_{\rA\to\rB_1}-\Theta_{\rA\to\rB_1}^{(1)}\|_\diamond+\|\Xi_{\rA\to\rB_2}-\Theta_{\rA\to\rB_2}^{(2)}\|_\diamond
    \end{equation}
    where $\Theta^{(1)}_{\rA\to\rB_1}:=\Tr_{\rB_2}\circ\Theta_{\rA\to\rB_1\rB_2}$ and $\Theta^{(2)}_{\rA\to\rB_2}:=\Tr_{\rB_1}\circ\Theta_{\rA\to\rB_1\rB_2}$ (see Fig.~\ref{fig:joint_realization}).
    The tightness of inequalities in Eq.~\eqref{eq:RoI-error} is explained in the following subsection.
    \begin{figure}
        \centering
        \includegraphics[width=0.5\linewidth]{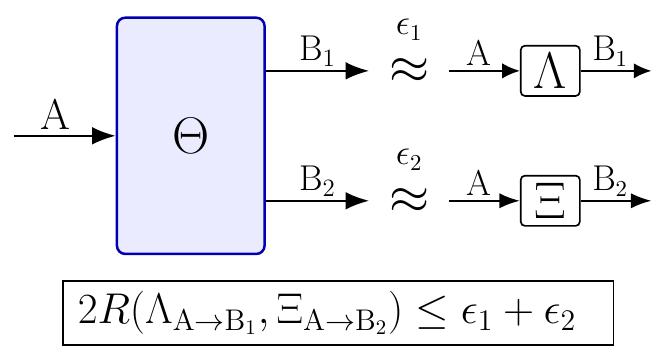}
        \caption{Approximate joint realization tradeoff of two quantum channels. A joint realization channel $\Theta_{\rA\to\rB_1\rB_2}$ induces two marginal channels $\Theta^{(1)}_{\mathrm A\to\rB_1}$ and $\Theta^{(2)}_{\mathrm A\to\rB_2}$, which approximate the target channels $\Lambda_{\mathrm A\to\mathrm B_1}$ and $\Xi_{\mathrm A\to\mathrm B_2}$ with errors $\epsilon_1:=\|\Lambda_{\rA\to\rB_1}-\Theta_{\rA\to\rB_1}^{(1)}\|_\diamond$ and $\epsilon_2:=\|\Xi_{\rA\to\rB_2}-\Theta_{\rA\to\rB_2}^{(2)}\|_\diamond$ respectively. Theorem~\ref{thm:UR_GR} shows $2R(\Lambda_{\mathrm A\to\mathrm B_1},\Xi_{\mathrm A\to\mathrm B_2})\le \epsilon_1+\epsilon_2$.}
        \label{fig:joint_realization}
    \end{figure}

    \subsection{Relationship with optimal symmetric universal cloning}\label{section:cloning}
    According to the no-broadcasting theorem~\cite{barnum1996noncommuting}, identity channels are incompatible.
    The following proposition provides the RoI of $n$ (a natural number greater than or equal to $2$) identity channels.
    \begin{proposition}\label{proposition:id_n_robustness}
        Let $\rA$ be a $d$-dimensional system and $\rA_i$ $(i=1,\dots,n)$ be a system with the same dimension as $\rA$.
        Then, the RoI becomes as follows:
        \begin{equation}\label{eq:RoI-identity}
            R(\{\rmid^{(i)}_{\rA\to\rA_i}\}_{i=1}^{n})=\frac{(n-1)(d-1)}{n+d-1}.
        \end{equation}
    \end{proposition}
    When $n=2$, the RoI is $(d-1)/(d+1)$, which coincides with the result of Ref.~\cite[Theorem 5]{Haapasalo2015}.

    According to the result of optimal symmetric universal cloning~\cite{keyl1999optimal}, when one approximately broadcasts a pure quantum state $\ket{\psi}_\rA$ using a channel from a single $d$-dimensional system to $n$ $d$-dimensional systems with the same fidelity, the marginal channel that provides the maximum fidelity is given by  
    \begin{equation}\label{eq:optimal_cloning}
        \cC_{\rA\to\rA_i}=\eta\rmid_{\rA\to\rA_i}+(1-\eta)\cD_{\rA\to\rA_i}
    \end{equation}
    where $\eta:=\frac{n+d}{n(d+1)}$ and $\cD_{\rA\to\rA_i}(X_\rA)=\Tr[X_\rA]\frac{\one_{\rA_i}}{d}$ is a depolarizing channel.
    Therefore, channels $\cC_{\rA\to\rA_1},\dots,\cC_{\rA\to\rA_n}$ are compatible.
    
    The following proposition shows the minimal total error of $n$ identity channels (for the proof, see Appendix~\ref{appendix:minimum_total_error_id}).
    \begin{proposition}\label{proposition:minimum_total_error_id}
    Let $\rA$ be a $d$-dimensional system and $\rA_i$ $(i=1,\dots,n)$ be a system with the same dimension as $\rA$.
    Then, the minimum total error is
    \begin{equation}\label{eq:identity_error}
        E(\{\rmid^{(i)}_{\rA\to\rA_i}\}_{i=1}^{n})=\frac{2(n-1)(d-1)}{d},
    \end{equation}
    and this is achieved by the optimal symmetric universal cloner whose marginal channels are given by Eq.~\eqref{eq:optimal_cloning}.
    \end{proposition}

    Now, we substitute the RoI of $n$ identity channels given by Eq.~\eqref{eq:RoI-identity} for our upper bound of Eq.~\eqref{eq:RoI-error}:
    \begin{equation}
        \frac{2nR(\{\rmid^{(i)}_{\rA\to\rA_i}\}_{i=1}^{n})}{1+R(\{\rmid^{(i)}_{\rA\to\rA_i}\}_{i=1}^{n})}=\frac{2n\frac{(n-1)(d-1)}{n+d-1}}{1+\frac{(n-1)(d-1)}{n+d-1}}=\frac{2(n-1)(d-1)}{d},
    \end{equation}
    which is exactly equal to $E(\{\rmid^{(i)}_{\rA\to\rA_i}\}_{i=1}^{n})$ given in Eq.~\eqref{eq:identity_error}.
    This means that our upper bound is tight.
    On the other hand, our lower bound of Eq.~\eqref{eq:RoI-error} is
    \begin{equation}
        2R(\{\rmid^{(i)}_{\rA\to\rA_i}\}_{i=1}^{n})=\frac{2(n-1)(d-1)}{n+d-1},
    \end{equation}
    which is strictly smaller than $E(\{\rmid^{(i)}_{\rA\to\rA_i}\}_{i=1}^{n})$.
    However, since we have
    \begin{equation}
        \lim_{d\to\infty}\frac{E(\{\rmid^{(i)}_{\rA\to\rA_i}\}_{i=1}^{n})}{2R(\{\rmid^{(i)}_{\rA\to\rA_i}\}_{i=1}^{n})}=\lim_{d\to\infty}\frac{n+d-1}{d}=1,
    \end{equation}
    the coefficient $2$ cannot be universally improved.

    \section{Measurement uncertainty relation}
	Applying this tradeoff to measurement channels, we can systematically derive a \emph{measurement uncertainty relation (MUR)}.
	When an incompatible POVM pair $(\E_\rA^\cX, \F_\rA^\cY)$ is approximated by a POVM $\G^{\cX\cY}_\rA$ with marginals $\G_\rA^{(1)\cX}$ and $\G_\rA^{(2)\cY}$, typical MURs are of the form $C\le d(\E_\rA^\cX, \G_\rA^{(1)\cX})\cdot d(\F_\rA^\cY, \G_\rA^{(2)\cY})$ \cite{Arthurs1965,Arthurs1988,Busch1984,ozawa2003universally,Werner2004,Busch2007,Busch2007a,watanabe2011uncertainty} or $C\le d(\E_\rA^\cX, \G_\rA^{(1)\cX})+d(\F_\rA^\cY, \G_\rA^{(2)\cY})$ \cite{Busch1986,Busch2008,Miyadera2008,Buscemi2014,Takakura2020}, with a certain ``distance'' $d(\cdot, \cdot)$ (called an error) between POVMs.
	The positive constant $C$ depends solely on the pair $(\E_\rA^\cX, \F_\rA^\cY)$, physically quantifying their joint unrealizability.
	Theorem \ref{thm:UR_GR} naturally yields an MUR where this constant is given by the RoI, recasting joint unmeasurability as a direct consequence of channel incompatibility.
	\begin{theorem}\label{theorem:mur}
		Let $\E^\cX_\rA,\F^\cY_\rA\in\sM(\rA)$ be POVMs.
		For any approximate joint POVM $\G^{\cX\cY}_\rA\in\sM(\rA)$ and its marginals $\G_\rA^{(1)\cX}$ and $\G_\rA^{(2)\cY}$, the following inequality holds:
		\begin{equation}
			\begin{split}
				2R\left(\E^\cX_\rA,\F^\cY_\rA\right)\le \epsilon\left(\E^\cX_\rA,\G^{(1)\cX}_\rA\right)+\epsilon\left(\F^\cY_\rA,\G^{(2)\cY}_\rA\right).
			\end{split}
		\end{equation}
		Here, $R\left(\E^\cX_\rA,\F^\cY_\rA\right)$ is the RoI of POVMs defined in Eq.~\eqref{eq:robustness_POVMs}, and we define an error between POVMs $\E^\cX_\rA$ and $\G^\cX_\rA$ by the $\ell_1$-type distance:
        \begin{equation}
            \epsilon(\E^\cX_\rA,\G^\cX_\rA):=\sum_{x\in\cX}\|E^x_\rA-G^x_\rA\|.
        \end{equation}
	\end{theorem}
	The proof is given in Appendix \ref{Appendix:Joint-realization tradeoffs from channel incompatibility}, where the error is derived from the diamond norm of the measurement channels.
	The error is faithful, i.e., $\epsilon(\E^\cX_\rA,\G^\cX_\rA)=0$ if and only if $\E^\cX_\rA=\G^\cX_\rA$, which is an analogue of Ozawa’s soundness and completeness requirements \cite{ozawa2019soundness} in a state-independent way.
    We mention that a related distance-based quantifier of measurement incompatibility, defined through the minimum distance to the set of jointly measurable POVMs, has been studied in Ref.~\cite{tendick2023distance}.
	
	While the operational meaning of RoI has been given by advantages in certain discrimination tasks \cite{Skrzypczyk2019,uola2019quantifying,mori2020operational,Uola2020,Ducuara2025}, its connection to ``joint unrealizability" has remained obscure. 
	Moreover, RoI coincides with the advantage only when discriminating a specifically tailored ensemble; for general ensembles, it merely serves as an upper bound.
	In contrast, our inequality endows RoI with an operational meaning directly rooted in joint unrealizability. 
	Theorem~\ref{thm:UR_GR} and Theorem~\ref{theorem:mur} establish RoI not just as an abstract resource measure but as a fundamental bound governing uncertainty in joint realizations of incompatible devices, distinguishing our framework from distance-based resource theories of incompatibility \cite{Mitra2025,Ghai2025}.
	
	\section{Information-error-disturbance tradeoff} 
	Applying Theorem \ref{thm:UR_GR}, we can derive an \emph{information-error-disturbance tradeoff} for implementing a target POVM $\E^\cX_\rA$ via a measurement process.
    To see this, we first introduce a quantifier of measurement informativeness \cite{skrzypczyk2019robustness}.
    \begin{definition}
        For a POVM $\E^\cX_\rA\in\sM(\rA)$, the robustness of measurement of $\E^\cX_\rA$  is defined as follows:
		\begin{align*}
			I(\E^\cX_\rA):=\mathrm{min}\left\{r\ge0\mid\N^\cX_\rA\in\sM(\rA),q_x\ge 0,\sum_{x\in\cX}q_x=1~\mathrm{s.t.}\;\frac{E^x_\rA+r N^x_\rA}{1+r}=q_x\one_\rA\;(\forall x\in\cX)\right\}.
		\end{align*}
    \end{definition}
    
    An uninformative measurement outputs the same distribution for any state, extracting no information from the system.
    Thus, the quantity $I(\mathsf{E}_A^\mathcal{X})$ captures the informativeness of $\mathsf{E}_A^\mathcal{X}$, the ability of the POVM to extract information.
    This measure of informativeness provides the lower bound of the degree of incompatibility between a measurement and the identity channel (see Appendix \ref{Appendix:information-error-disturbance tradeoff} for the proof).
    \begin{proposition}\label{proposition:RoM_vs_RoI}
    	Let $\E^\cX_\rA\in\sM(\rA)$ be a POVM and $d:=\dim\cH_\rA(\ge 2)$.
    	It holds that
    	\begin{equation}\label{eq:RvsRoI}
    		\frac{\left(\sqrt{I(\E^\cX_\rA)+1}-1\right)^2}{d-1}\le R(\rmid_\rA,\E^\cX_\rA).
    	\end{equation}
    \end{proposition}
    Since the identity channel corresponds to preserving the system, its incompatibility with a POVM expresses the disturbance caused by the measurement.
    This proposition thus serves as a resource-theoretic representation of the no-information without disturbance principle; it quantitatively shows that extracting more information through a measurement forces stronger disturbance.

    Combining this observation with Theorem~\ref{thm:UR_GR} leads to an information-error-disturbance tradeoff.
    \begin{theorem}\label{theorem:i-d-e}
		Let $\E^\cX_\rA, \K^\cX_\rA\in\sM(\rA)$ be POVMs and $d:=\dim\cH_\rA(\ge 2)$.
		For any $\K^\cX_\rA$-channel $\Lambda^{\K}_{\rA\to\rB}\in\sC(\rA\to\rB)$, the following information-error-disturbance tradeoff holds:
			\begin{equation}\label{eq:error-disturbance}
				\frac{2\left(\sqrt{I(\E^\cX_\rA)+1}-1\right)^2}{d-1}\le \epsilon(\E^\cX_\rA,\K^\cX_\rA)+\delta(\Lambda^{\K}_{\rA\to\rB}).
			\end{equation}
            Here, a disturbance of a channel $\Lambda^{\K}_{\rA\to\rB}$ is defined as follows~\cite{kretschmann2008information, heinosaari2013qualitative}:
		\begin{equation}\label{eq:disturbance}
			\delta(\Lambda^{\K}_{\rA\to\rB}):=\inf_{\cR_{\rB\to\rA}\in\sC(\rB\to\rA)}\|\cR_{\rB\to\rA}\circ\Lambda^{\K}_{\rA\to\rB}-\rmid_\rA\|_\diamond.
		\end{equation}
	\end{theorem}
    \begin{proof}
        We consider measuring a target POVM $\E^\cX_\rA\in\sM(\rA)$ through a CP-instrument corresponding to a compatible channel-POVM pair $(\tilde{\Lambda}^\K_{\rA},\K^\cX_{\rA})$.
        By applying Theorem~\ref{thm:UR_GR} to the channel pair $(\rmid_\rA,\Gamma^{\E^\cX_\rA})$ and by analyzing the error in the same way as in Theorem~\ref{theorem:mur}, we have
        \begin{equation}\label{eq:RoI-d-e}
			2R(\rmid_\rA,\E^\cX_\rA)\le\epsilon(\E^\cX_\rA,\K^\cX_\rA)+\|\tilde{\Lambda}^\K_{\rA}-\rmid_\rA\|_\diamond,
		\end{equation}
        where $R(\rmid_\rA,\E^\cX_\rA)$ is the RoI defined for a channel and POVM defined in Eq.~\eqref{eq:RoI_channel-POVM}.
        Notice that if $\Lambda^\K_{\rA\to\rB}\in\sC(\rA\to\rB)$ is a $\K^\cX_\rA$-channel, for any CPTP map $\cR_{\rB\to\rA}\in\sC(\rB\to\rA)$, the composition $\cR_{\rB\to\rA}\circ\Lambda^\K_{\rA\to\rB}$ is a $\K^\cX_\rA$-channel from $\rA$ to $\rA$~\cite{heinosaari2013qualitative}.
		Since Eq.~\eqref{eq:RoI-d-e} holds for any $\K^\cX_\rA$-channel $\tilde{\Lambda}^\K_{\rA}\in\sC(\rA\to\rA)$, we have $2R(\rmid_\rA,\E^\cX_\rA)\le\epsilon(\E^\cX_\rA,\K^\cX_\rA)+\|\cR_{\rB\to\rA}\circ\Lambda^\K_{\rA\to\rB}-\rmid_\rA\|_\diamond$ for all $\cR_{\rB\to\rA}\in\sC(\rB\to\rA)$.
        Combining this inequality with Proposition~\ref{proposition:RoM_vs_RoI}, we obtain the desired inequality.
    \end{proof}
	
	While RoI quantifies disturbance abstractly through resource strength, the quantity $\delta$ captures operational irreversibility of the system under the optimal recovery $\mathcal{R}$.
	Theorem \ref{theorem:i-d-e} demonstrates that to approximate an informative measurement $\E^\cX_\rA$ with a measurement process $(\Lambda^\K_{\rA\to\rB},\K^\cX_\rA)$, one must balance the error $\epsilon(\E^\cX_\rA,\K^\cX_\rA)$ of the approximation and the disturbance $\delta(\Lambda^\K_{\rA\to\rB})$ to the target system.
    Setting $\K^\cX_\rA=\E^\cX_\rA$ in Eq.~\eqref{eq:error-disturbance} yields a pure information-disturbance tradeoff:
	\begin{corollary}   
		Let $\E^\cX_\rA\in\sM(\rA)$ be a POVM, $\Lambda^\E_{\rA\to\rB}\in\sC(\rA\to\rB)$ be an $\E^\cX_\rA$-channel, and $d:=\dim\cH_\rA(\ge 2)$.
		The following information-disturbance tradeoff relation holds:
		\begin{equation}\label{eq:inf-dist-robustness}
			\frac{2\left(\sqrt{I(\E^\cX_\rA)+1}-1\right)^2}{d-1}\le\delta(\Lambda^\E_{\rA\to\rB}).
		\end{equation}
	\end{corollary}

	Our relation provides a tighter lower bound on disturbance than known algebraic limits, such as the Heinosaari–Miyadera (HM) bound \cite{heinosaari2013qualitative}:
	\begin{equation}\label{eq:HM-bound}
		\begin{split}
			\frac{1}{16}\sup_{x\in\cX}(\|E_\rA^x\|+\|\one_\rA-E_\rA^x\|-1)^{2}
			\le\delta(\Lambda^\E_{\rA\to\rB}).
		\end{split}
	\end{equation}
	When the dimension $d$ of the system satisfies $d\le6$, 
		our bound dominates the HM bound (see Appendix \ref{Appendix:information-error-disturbance tradeoff}):
	\begin{equation}
		\frac{1}{16}\sup_{x\in\cX}(\|E_\rA^x\|+\|\one_\rA-E_\rA^x\|-1)^{2}\le \frac{2\Big(\sqrt{I(\E^\cX_\rA)+1}-1\Big)^{2}}{d-1}.
	\end{equation}
	
    A distinctive feature of our results is that while the HM bound was derived by using the continuity of Stinespring dilation~\cite{stinespring1955positive,kretschmann2008information}, our formulation does not invoke any dilation or complementary-channel representation and applies directly to arbitrary families of channels that need not arise as complementary parts of a common isometry.
    
    \begin{figure}[tbp]
	 	\centering
	 	\includegraphics[width=0.5\linewidth]{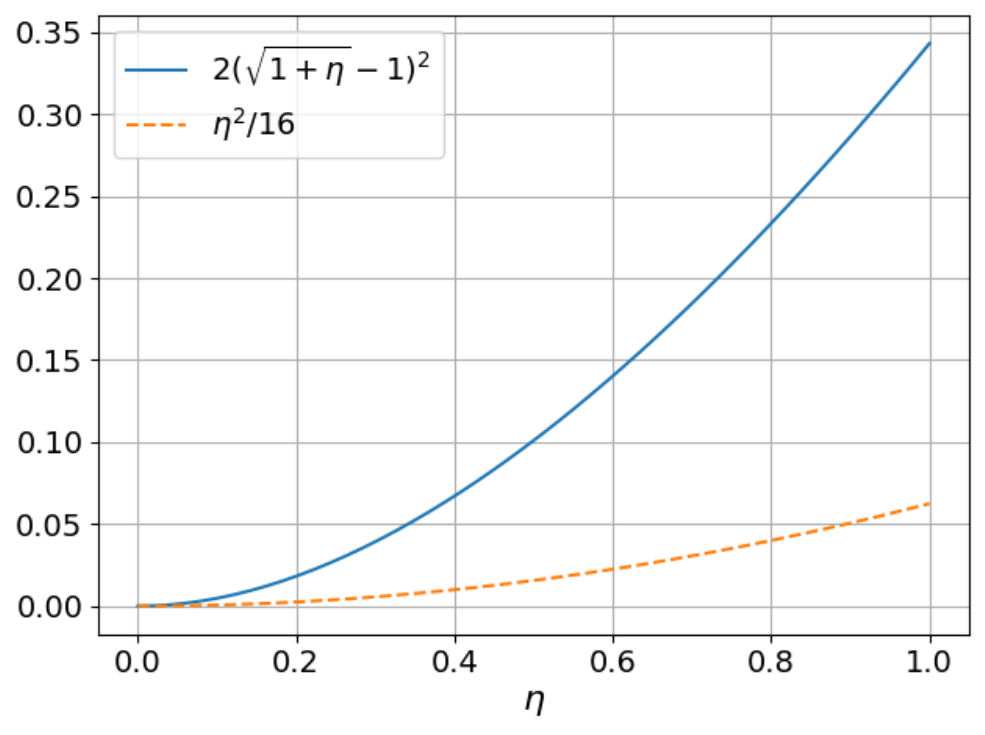}
	 	\caption{Comparison with the robustness-based bound $2(\sqrt{1+
	 			\eta}-1)^2$ and Heinosaari--Miyadera's bound $\eta^2/16$ for the unbiased qubit POVM  Eq.~\eqref{eq:def_unbiased_Pauli-Z}.}
	 	\label{fig:unbiased}
	\end{figure}
    \begin{example}\label{example:unbiased}
		We consider an unbiased binary qubit POVM \cite{Busch2009}  $\mathsf{Z}(\eta)_\rA:=\{Z(\eta)^j_\rA\}_{j=\pm1}$ parameterized by $\eta\in[0,1]$, where
		\begin{equation}\label{eq:def_unbiased_Pauli-Z}
			Z(\eta)^{\pm1}_\rA:=\frac{1}{2}(\one_\rA\pm\eta\sigma_\rA^z)
		\end{equation}
		and $\sigma^z_\rA$ is the Pauli-Z operator.
        Since $\|Z(\eta)^{\pm1}_\rA\|=\frac{1+\eta}{2}$, we have $I(\mathsf{Z}(\eta)_\rA)=\eta$.
		From Eq.~\eqref{eq:inf-dist-robustness}, it holds that $2(\sqrt{1+\eta}-1)^2\le \delta(\Lambda^{\mathsf{Z}(\eta)}_\rA)$, while the HM bound in \eq{eq:HM-bound} is $\eta^{2}/16$.
		Our bound outperforms the HM bound for all $\eta \in (0, 1]$ (see Fig.~\ref{fig:unbiased}).  
        Moreover, this example achieves the equality of Eq.~\eqref{eq:RvsRoI} (see Appendix \ref{Appendix:information-error-disturbance tradeoff}).
	\end{example}
	
	\section{Conclusion}
	By deriving the RoI-based lower and upper bound of the minimum total error of approximate joint realization of incompatible channels, we connected the traditional concept of joint unrealizability (uncertainty) with the modern resource theory of incompatibility. 
	This proves that the implementation error tradeoff for any approximate joint realization of incompatible quantum channels is quantitatively bounded by twice RoI of those channels. 
	Evaluating these channels as measurement or identity channels, we established a unified framework for error-error and information-disturbance tradeoffs. 
	Crucially, our results recast the resource-theoretic strength of incompatibility directly as uncertainty bounds, without relying on specific measurement models \cite{von1955mathematical,ozawa1984quantum} or algebraic frameworks such as the Kennard-Robertson inequality \cite{kennard1927zur,robertson1929,ozawa2003universally,ozawa2004uncertainty} or Stinespring's theorem~\cite{kretschmann2008information}.
	Extending these tradeoffs to incompatibility of quantum instruments \cite{Mitra2022,Mitra2023,Leppaejaervi2024} complements instrument-based resource theory \cite{buscemi2023unifying,Ji2024}.
	Several directions remain open. 
	While the current method fails for weight of incompatibility \cite{Uola2020a,Ducuara2025,SudarsananRagini2024} (see Appendix \ref{appendix:weight}), deriving tradeoffs from other quantifiers of incompatibility remains to be explored. 
	Furthermore, replacing our SDP-based optimization with analytically computable bounds is a future challenge.

    \section*{Data Availability}
    No datasets were generated or analyzed in this work. 
    The graphs in Figs.~\ref{fig:unbiased} and \ref{fig:povm_example} are obtained by directly evaluating the analytical expressions given in the main text and Appendix \ref{Appendix:information-error-disturbance tradeoff}.

	\section*{Acknowledgments}
	The authors thank the Yukawa Institute for Theoretical Physics at Kyoto University, where this work was initiated during the YITP-W-25-11.
    S.M. thanks Ion Nechita for insightful discussions.
	S.M. acknowledges support from the French government under the France 2030 investment plan, as part of the Initiative d'Excellence d'Aix-Marseille Université-A*MIDEX, AMX-22-CEI-01.
	R.T. thanks Takayuki Miyadera for insightful discussions.
	R.T. acknowledges support from JST COI-NEXT program Grant No. JPMJPF2014, and JSPS KAKENHI Grant No. JP25K17314.
	K.T. acknowledges support from JST SPRING Grant No. JPMJSP2125.
    
    \section*{Author Contributions}
    S.M.: Conceptualization, Formal analysis, Writing – original draft, Writing – review \& editing. 
    R.T.: Conceptualization, Formal analysis, Writing – review \& editing.
    K.T.: Formal analysis, Writing – review \& editing.
    OpenAI ChatGPT (GPT-5.6 Sol and GPT-6 Astra) was used for bibliographic research, refining sentences, and assistance with the proof strategy and analysis in Proposition~\ref{proposition:id_n_robustness} (especially showing Eq.~\eqref{eq:Q_ineqality}) and Proposition~\ref{proposition:minimum_total_error_id}.
    Earlier versions of ChatGPT were also used for mathematical discussions and for improving text; in particular, the proof of Proposition~\ref{proposition:RoM_vs_RoI} was significantly assisted by ChatGPT.
    The authors checked the output by AI and take full responsibility for the contents of the manuscript. 
    
    \bibliographystyle{quantum}
	\bibliography{Quantum/myref}

\begin{thebibliography}{10}

\bibitem{Busch2016}
Paul Busch, Pekka~J Lahti, Juha-Pekka Pellonp{\"a}{\"a}, and Kari Ylinen.
\newblock ``Quantum measurement''.
\newblock Springer International Publishing. ~(2016).

\bibitem{guhne2023colloquium}
Otfried G\"uhne, Erkka Haapasalo, Tristan Kraft, Juha-Pekka Pellonp\"a\"a, and Roope Uola.
\newblock ``Colloquium: Incompatible measurements in quantum information science''.
\newblock \href{https://dx.doi.org/10.1103/RevModPhys.95.011003}{Rev. Mod. Phys. {\bf 95}, 011003}~(2023).

\bibitem{busch2009no}
Paul Busch.
\newblock ``{``No Information Without Disturbance'': Quantum Limitations of Measurement}''.
\newblock In Wayne~C. Myrvold and Joy Christian, editors, {Quantum Reality, Relativistic Causality, and Closing the Epistemic Circle: Essays in Honour of Abner Shimony}.
\newblock \href{https://dx.doi.org/10.1007/978-1-4020-9107-0_13}{Pages 229--256}.
\newblock Springer Netherlands, Dordrecht~(2009).

\bibitem{dariano2017quantum}
Giacomo~Mauro D'Ariano, Giulio Chiribella, and Paolo Perinotti.
\newblock ``Quantum theory from first principles: An informational approach''.
\newblock \href{https://dx.doi.org/10.1017/9781107338340}{Cambridge University Press, Cambridge, UK}. ~(2017).

\bibitem{Wootters1982}
William~K Wootters and Wojciech~H Zurek.
\newblock ``A single quantum cannot be cloned''.
\newblock \href{https://dx.doi.org/https://doi.org/10.1038/299802a0}{Nature {\bf 299}, 802--803}~(1982).

\bibitem{Dieks1982}
D.~Dieks.
\newblock ``Communication by {EPR} devices''.
\newblock \href{https://dx.doi.org/https://doi.org/10.1016/0375-9601(82)90084-6}{Phys. Lett. A {\bf 92}, 271--272}~(1982).

\bibitem{Wiesner1983}
Stephen Wiesner.
\newblock ``Conjugate coding''.
\newblock \href{https://dx.doi.org/10.1145/1008908.1008920}{SIGACT News {\bf 15}, 78–88}~(1983).

\bibitem{barnum1996noncommuting}
Howard Barnum, Carlton~M. Caves, Christopher~A. Fuchs, Richard Jozsa, and Benjamin Schumacher.
\newblock ``Noncommuting mixed states cannot be broadcast''.
\newblock \href{https://dx.doi.org/10.1103/PhysRevLett.76.2818}{Phys. Rev. Lett. {\bf 76}, 2818--2821}~(1996).

\bibitem{barnum2007generalized}
Howard Barnum, Jonathan Barrett, Matthew Leifer, and Alexander Wilce.
\newblock ``Generalized no-broadcasting theorem''.
\newblock \href{https://dx.doi.org/10.1103/PhysRevLett.99.240501}{Phys. Rev. Lett. {\bf 99}, 240501}~(2007).

\bibitem{heinosaari2017incompatibility}
Teiko Heinosaari and Takayuki Miyadera.
\newblock ``Incompatibility of quantum channels''.
\newblock \href{https://dx.doi.org/10.1088/1751-8121/aa5f6b}{J. Phys. A: Math. Theor. {\bf 50}, 135302}~(2017).

\bibitem{heinosaari2016invitation}
Teiko Heinosaari, Takayuki Miyadera, and Mário Ziman.
\newblock ``An invitation to quantum incompatibility''.
\newblock \href{https://dx.doi.org/10.1088/1751-8113/49/12/123001}{J. Phys. A: Math. Theor. {\bf 49}, 123001}~(2016).

\bibitem{heisenberg1927uber}
Werner~Karl Heisenberg.
\newblock ``{\"U}ber den anschaulichen inhalt der quantentheoretischen kinematik und mechanik''.
\newblock \href{https://dx.doi.org/10.1007/BF01397280}{Zeitschrift f{\"u}r Physik {\bf 43}, 172--198}~(1927).

\bibitem{ozawa2003universally}
Masanao Ozawa.
\newblock ``Universally valid reformulation of the {H}eisenberg uncertainty principle on noise and disturbance in measurement''.
\newblock \href{https://dx.doi.org/10.1103/PhysRevA.67.042105}{Phys. Rev. A {\bf 67}, 042105}~(2003).

\bibitem{ozawa2004uncertainty}
Masanao Ozawa.
\newblock ``Uncertainty relations for joint measurements of noncommuting observables''.
\newblock \href{https://dx.doi.org/https://doi.org/10.1016/j.physleta.2003.12.001}{Phys. Lett. A {\bf 320}, 367--374}~(2004).

\bibitem{watanabe2011uncertainty}
Yu~Watanabe, Takahiro Sagawa, and Masahito Ueda.
\newblock ``Uncertainty relation revisited from quantum estimation theory''.
\newblock \href{https://dx.doi.org/10.1103/PhysRevA.84.042121}{Phys. Rev. A {\bf 84}, 042121}~(2011).

\bibitem{busch2013proof}
Paul Busch, Pekka Lahti, and Reinhard~F. Werner.
\newblock ``Proof of {H}eisenberg's error-disturbance relation''.
\newblock \href{https://dx.doi.org/10.1103/PhysRevLett.111.160405}{Phys. Rev. Lett. {\bf 111}, 160405}~(2013).

\bibitem{branciard2013error}
Cyril Branciard.
\newblock ``Error-tradeoff and error-disturbance relations for incompatible quantum measurements''.
\newblock \href{https://dx.doi.org/10.1073/pnas.1219331110}{Proc. Natl. Acad. Sci. {\bf 110}, 6742--6747}~(2013).

\bibitem{Buscemi2014}
Francesco Buscemi, Michael J.~W. Hall, Masanao Ozawa, and Mark~M. Wilde.
\newblock ``Noise and disturbance in quantum measurements: An information-theoretic approach''.
\newblock \href{https://dx.doi.org/10.1103/PhysRevLett.112.050401}{Phys. Rev. Lett. {\bf 112}, 050401}~(2014).

\bibitem{maccone2006information}
Lorenzo Maccone.
\newblock ``Information-disturbance tradeoff in quantum measurements''.
\newblock \href{https://dx.doi.org/10.1103/PhysRevA.73.042307}{Phys. Rev. A {\bf 73}, 042307}~(2006).

\bibitem{buscemi2006information}
Francesco Buscemi and Massimiliano~F. Sacchi.
\newblock ``Information-disturbance trade-off in quantum-state discrimination''.
\newblock \href{https://dx.doi.org/10.1103/PhysRevA.74.052320}{Phys. Rev. A {\bf 74}, 052320}~(2006).

\bibitem{maccone2007entropic}
L.~Maccone.
\newblock ``Entropic information-disturbance tradeoff''.
\newblock \href{https://dx.doi.org/10.1209/0295-5075/77/40002}{Europhys. Lett. {\bf 77}, 40002}~(2007).

\bibitem{kretschmann2008information}
Dennis Kretschmann, Dirk Schlingemann, and Reinhard~F. Werner.
\newblock ``The information-disturbance tradeoff and the continuity of {S}tinespring's representation''.
\newblock \href{https://dx.doi.org/10.1109/TIT.2008.917696}{IEEE Trans. Inf. Theory {\bf 54}, 1708--1717}~(2008).

\bibitem{heinosaari2013qualitative}
Teiko Heinosaari and Takayuki Miyadera.
\newblock ``Qualitative noise-disturbance relation for quantum measurements''.
\newblock \href{https://dx.doi.org/10.1103/PhysRevA.88.042117}{Phys. Rev. A {\bf 88}, 042117}~(2013).

\bibitem{dariano2020information}
Giacomo~Mauro D'Ariano, Paolo Perinotti, and Alessandro Tosini.
\newblock ``Information and disturbance in operational probabilistic theories''.
\newblock \href{https://dx.doi.org/10.22331/q-2020-11-16-363}{{Quantum} {\bf 4}, 363}~(2020).

\bibitem{Haapasalo2015}
Erkka Haapasalo.
\newblock ``Robustness of incompatibility for quantum devices''.
\newblock \href{https://dx.doi.org/10.1088/1751-8113/48/25/255303}{J. Phys. A: Math. Theor. {\bf 48}, 255303}~(2015).

\bibitem{Skrzypczyk2019}
Paul Skrzypczyk, Ivan \ifmmode \check{S}\else \v{S}\fi{}upi\ifmmode~\acute{c}\else \'{c}\fi{}, and Daniel Cavalcanti.
\newblock ``All sets of incompatible measurements give an advantage in quantum state discrimination''.
\newblock \href{https://dx.doi.org/10.1103/PhysRevLett.122.130403}{Phys. Rev. Lett. {\bf 122}, 130403}~(2019).

\bibitem{Carmeli2019}
Claudio Carmeli, Teiko Heinosaari, and Alessandro Toigo.
\newblock ``Quantum incompatibility witnesses''.
\newblock \href{https://dx.doi.org/10.1103/PhysRevLett.122.130402}{Phys. Rev. Lett. {\bf 122}, 130402}~(2019).

\bibitem{uola2019quantifying}
Roope Uola, Tristan Kraft, Jiangwei Shang, Xiao-Dong Yu, and Otfried G\"uhne.
\newblock ``Quantifying quantum resources with conic programming''.
\newblock \href{https://dx.doi.org/10.1103/PhysRevLett.122.130404}{Phys. Rev. Lett. {\bf 122}, 130404}~(2019).

\bibitem{mori2020operational}
Junki Mori.
\newblock ``Operational characterization of incompatibility of quantum channels with quantum state discrimination''.
\newblock \href{https://dx.doi.org/10.1103/PhysRevA.101.032331}{Phys. Rev. A {\bf 101}, 032331}~(2020).

\bibitem{Uola2020}
Roope Uola, Tristan Kraft, and Alastair~A. Abbott.
\newblock ``Quantification of quantum dynamics with input-output games''.
\newblock \href{https://dx.doi.org/10.1103/PhysRevA.101.052306}{Phys. Rev. A {\bf 101}, 052306}~(2020).

\bibitem{Ducuara2025}
Andr\'es~F. Ducuara, Ryo Takakura, Fernando~J. Hernandez, and Cristian~E. Susa.
\newblock ``Multiobject operational tasks for measurement incompatibility''.
\newblock \href{https://dx.doi.org/10.1103/m7ln-tb1s}{Phys. Rev. Res. {\bf 7}, 033050}~(2025).

\bibitem{kitaev1997quantum}
A~Yu Kitaev.
\newblock ``Quantum computations: algorithms and error correction''.
\newblock \href{https://dx.doi.org/10.1070/RM1997v052n06ABEH002155}{Russian Mathematical Surveys {\bf 52}, 1191}~(1997).

\bibitem{keyl1999optimal}
M.~Keyl and R.~F. Werner.
\newblock ``Optimal cloning of pure states, testing single clones''.
\newblock \href{https://dx.doi.org/10.1063/1.532887}{Journal of Mathematical Physics {\bf 40}, 3283--3299}~(1999).

\bibitem{vidal1999robustness}
Guifr\'e Vidal and Rolf Tarrach.
\newblock ``Robustness of entanglement''.
\newblock \href{https://dx.doi.org/10.1103/PhysRevA.59.141}{Phys. Rev. A {\bf 59}, 141--155}~(1999).

\bibitem{Steiner2003}
Michael Steiner.
\newblock ``Generalized robustness of entanglement''.
\newblock \href{https://dx.doi.org/10.1103/PhysRevA.67.054305}{Phys. Rev. A {\bf 67}, 054305}~(2003).

\bibitem{datta2009max}
Nilanjana Datta.
\newblock ``Max-{R}elative {E}ntropy of {E}ntanglement, {A}lias {L}og {R}obustness''.
\newblock \href{https://dx.doi.org/10.1142/S0219749909005298}{Int. J. Quantum Inf. {\bf 07}, 475--491}~(2009).

\bibitem{skrzypczyk2019robustness}
Paul Skrzypczyk and Noah Linden.
\newblock ``Robustness of measurement, discrimination games, and accessible information''.
\newblock \href{https://dx.doi.org/10.1103/PhysRevLett.122.140403}{Phys. Rev. Lett. {\bf 122}, 140403}~(2019).

\bibitem{oszmaniec2019operational}
Micha{\l{}} Oszmaniec and Tanmoy Biswas.
\newblock ``Operational relevance of resource theories of quantum measurements''.
\newblock \href{https://dx.doi.org/10.22331/q-2019-04-26-133}{{Quantum} {\bf 3}, 133}~(2019).

\bibitem{takagi2019general}
Ryuji Takagi and Bartosz Regula.
\newblock ``General resource theories in quantum mechanics and beyond: Operational characterization via discrimination tasks''.
\newblock \href{https://dx.doi.org/10.1103/PhysRevX.9.031053}{Phys. Rev. X {\bf 9}, 031053}~(2019).

\bibitem{chitambar2019quantum}
Eric Chitambar and Gilad Gour.
\newblock ``Quantum resource theories''.
\newblock \href{https://dx.doi.org/10.1103/RevModPhys.91.025001}{Rev. Mod. Phys. {\bf 91}, 025001}~(2019).

\bibitem{gour2025quantum}
Gilad Gour.
\newblock ``Quantum resource theories''.
\newblock \href{https://dx.doi.org/10.1017/9781009560870}{Cambridge University Press, Cambridge, UK}. ~(2025).

\bibitem{nielsen2010quantum}
Michael~A. Nielsen and Isaac~L. Chuang.
\newblock ``Quantum computation and quantum information: 10th anniversary edition''.
\newblock \href{https://dx.doi.org/10.1017/CBO9780511976667}{Cambridge University Press, Cambridge, UK}. ~(2010).

\bibitem{wilde2017quantum}
Mark~M. Wilde.
\newblock ``Quantum information theory''.
\newblock \href{https://dx.doi.org/https://doi.org/10.1017/9781316809976}{Cambridge University Press, Cambridge, UK}. ~(2017).
\newblock 2nd edition.

\bibitem{watrous2018theory}
John Watrous.
\newblock ``The theory of quantum information''.
\newblock \href{https://dx.doi.org/10.1017/9781316848142}{Cambridge University Press, Cambridge, UK}. ~(2018).

\bibitem{Holevo2020}
A~S Holevo and A~A Kuznetsova.
\newblock ``Information capacity of continuous variable measurement channel''.
\newblock \href{https://dx.doi.org/10.1088/1751-8121/ab7df8}{J. Phys. A: Math. Theor. {\bf 53}, 175304}~(2020).

\bibitem{ozawa1984quantum}
Masanao Ozawa.
\newblock ``Quantum measuring processes of continuous observables''.
\newblock \href{https://dx.doi.org/10.1063/1.526000}{J. Math. Phys. {\bf 25}, 79--87}~(1984).

\bibitem{watrous2009semidefinite}
John Watrous.
\newblock ``Semidefinite programs for completely bounded norms''.
\newblock \href{https://dx.doi.org/10.4086/toc.2009.v005a011}{Theory of Computing {\bf 5}, 217--238}~(2009).

\bibitem{ben-aroya2010on}
Avraham Ben-Aroya and Amnon Ta-Shma.
\newblock ``On the complexity of approximating the diamond norm''.
\newblock \href{https://dx.doi.org/10.26421/QIC10.1-2-6}{Quantum Information and Computation {\bf 10}, 77--86}~(2010).

\bibitem{watrous2013simpler}
John Watrous.
\newblock ``Simpler semidefinite programs for completely bounded norms''.
\newblock \href{https://dx.doi.org/10.4086/cjtcs.2013.008}{Chicago J. Theor. Comput. Sci. {\bf 2013}, 8}~(2013).

\bibitem{boyd2004convex}
Stephen Boyd and Lieven Vandenberghe.
\newblock ``Convex optimization''.
\newblock \href{https://dx.doi.org/10.1017/CBO9780511804441}{Cambridge University Press, Cambridge, UK}. ~(2004).

\bibitem{Heinosaari2015}
Teiko Heinosaari, Jukka Kiukas, and Daniel Reitzner.
\newblock ``Noise robustness of the incompatibility of quantum measurements''.
\newblock \href{https://dx.doi.org/10.1103/PhysRevA.92.022115}{Phys. Rev. A {\bf 92}, 022115}~(2015).

\bibitem{Uola2015}
Roope Uola, Costantino Budroni, Otfried G\"uhne, and Juha-Pekka Pellonp\"a\"a.
\newblock ``One-to-one mapping between steering and joint measurability problems''.
\newblock \href{https://dx.doi.org/10.1103/PhysRevLett.115.230402}{Phys. Rev. Lett. {\bf 115}, 230402}~(2015).

\bibitem{Designolle2019}
Sébastien Designolle, Máté Farkas, and Jędrzej Kaniewski.
\newblock ``Incompatibility robustness of quantum measurements: a unified framework''.
\newblock \href{https://dx.doi.org/10.1088/1367-2630/ab5020}{New J. Phys. {\bf 21}, 113053}~(2019).

\bibitem{Guerini2017}
Leonardo Guerini, Jessica Bavaresco, Marcelo Terra~Cunha, and Antonio Ac{í}n.
\newblock ``Operational framework for quantum measurement simulability''.
\newblock \href{https://dx.doi.org/10.1063/1.4994303}{J. Math. Phys. {\bf 58}, 092102}~(2017).

\bibitem{Arthurs1965}
E.~Arthurs and J.~L. Kelly~Jr.
\newblock ``On the simultaneous measurement of a pair of conjugate observables''.
\newblock \href{https://dx.doi.org/10.1002/j.1538-7305.1965.tb01684.x}{Bell Syst. Tech. J. {\bf 44}, 725--729}~(1965).

\bibitem{Arthurs1988}
E.~Arthurs and M.~S. Goodman.
\newblock ``Quantum correlations: A generalized {H}eisenberg uncertainty relation''.
\newblock \href{https://dx.doi.org/10.1103/PhysRevLett.60.2447}{Phys. Rev. Lett. {\bf 60}, 2447--2449}~(1988).

\bibitem{Busch1984}
Paul Busch and Pekka~J. Lahti.
\newblock ``On various joint measurements of position and momentum observables in quantum theory''.
\newblock \href{https://dx.doi.org/10.1103/PhysRevD.29.1634}{Phys. Rev. D {\bf 29}, 1634--1646}~(1984).

\bibitem{Werner2004}
Reinhard~F. Werner.
\newblock ``The uncertainty relation for joint measurement of position and momentum''.
\newblock \href{https://dx.doi.org/10.26421/QIC4.6-7-13}{Quantum Info. Comput. {\bf 4}, 546–562}~(2004).

\bibitem{Busch2007}
P.~Busch and D.~B. Pearson.
\newblock ``Universal joint-measurement uncertainty relation for error bars''.
\newblock \href{https://dx.doi.org/10.1063/1.2759831}{J. Math. Phys. {\bf 48}, 082103}~(2007).

\bibitem{Busch2007a}
Paul Busch, Teiko Heinonen, and Pekka Lahti.
\newblock ``{H}eisenberg's uncertainty principle''.
\newblock \href{https://dx.doi.org/https://doi.org/10.1016/j.physrep.2007.05.006}{Phys. Rep. {\bf 452}, 155 -- 176}~(2007).

\bibitem{Busch1986}
Paul Busch.
\newblock ``Unsharp reality and joint measurements for spin observables''.
\newblock \href{https://dx.doi.org/10.1103/PhysRevD.33.2253}{Phys. Rev. D {\bf 33}, 2253--2261}~(1986).

\bibitem{Busch2008}
Paul Busch and Teiko Heinosaari.
\newblock ``Approximate joint measurements of qubit observables''.
\newblock \href{https://dx.doi.org/10.26421/QIC8.8-9-9}{Quantum Inform. Compu. {\bf 8}, 797–818}~(2008).

\bibitem{Miyadera2008}
Takayuki Miyadera and Hideki Imai.
\newblock ``{H}eisenberg's uncertainty principle for simultaneous measurement of positive-operator-valued measures''.
\newblock \href{https://dx.doi.org/10.1103/PhysRevA.78.052119}{Phys. Rev. A {\bf 78}, 052119}~(2008).

\bibitem{Takakura2020}
Ryo Takakura and Takayuki Miyadera.
\newblock ``Preparation uncertainty implies measurement uncertainty in a class of generalized probabilistic theories''.
\newblock \href{https://dx.doi.org/10.1063/5.0017854}{J. Math. Phys. {\bf 61}, 082203}~(2020).

\bibitem{ozawa2019soundness}
Masanao Ozawa.
\newblock ``Soundness and completeness of quantum root-mean-square errors''.
\newblock \href{https://dx.doi.org/10.1038/s41534-018-0113-z}{npj Quantum Inf. {\bf 5}, 1}~(2019).

\bibitem{tendick2023distance}
Lucas Tendick, Martin Kliesch, Hermann Kampermann, and Dagmar Bru{\ss{}}.
\newblock ``Distance-based resource quantification for sets of quantum measurements''.
\newblock \href{https://dx.doi.org/10.22331/q-2023-05-15-1003}{{Quantum} {\bf 7}, 1003}~(2023).

\bibitem{Mitra2025}
Arindam Mitra, Sumit Mukherjee, and Changhyoup Lee.
\newblock ``Distance-based measures and epsilon measures for measurement-based quantum resources''.
\newblock \href{https://dx.doi.org/10.1103/wwn7-j8df}{Phys. Rev. A {\bf 112}, 062233}~(2025).

\bibitem{Ghai2025}
Jatin Ghai and Arindam Mitra.
\newblock ``Instrument-based quantum resources: quantification, hierarchies and constructing resource theories''~(2025).
\newblock  \href{http://arxiv.org/abs/2508.09134}{arXiv:2508.09134}.

\bibitem{stinespring1955positive}
W.~Forrest Stinespring.
\newblock ``Positive functions on {C}*-algebras''.
\newblock \href{https://dx.doi.org/10.2307/2032342}{Proceedings of the American Mathematical Society {\bf 6}, 211--216}~(1955).

\bibitem{Busch2009}
Paul Busch.
\newblock ``On the sharpness and bias of quantum effects''.
\newblock \href{https://dx.doi.org/https://doi.org/10.1007/s10701-009-9287-8}{Found. Phys. {\bf 39}, 712--730}~(2009).

\bibitem{von1955mathematical}
John von Neumann.
\newblock ``Mathematical foundations of quantum mechanics''.
\newblock Princeton University Press, Princeton, NJ. ~(1955).

\bibitem{kennard1927zur}
E.~H. Kennard.
\newblock ``Zur quantenmechanik einfacher bewegungstypen''.
\newblock \href{https://dx.doi.org/10.1007/BF01391200}{Zeitschrift f{\"u}r Physik {\bf 44}, 326--352}~(1927).

\bibitem{robertson1929}
H.~P. Robertson.
\newblock ``The uncertainty principle''.
\newblock \href{https://dx.doi.org/10.1103/PhysRev.34.163}{Phys. Rev. {\bf 34}, 163--164}~(1929).

\bibitem{Mitra2022}
Arindam Mitra and M\'at\'e Farkas.
\newblock ``Compatibility of quantum instruments''.
\newblock \href{https://dx.doi.org/10.1103/PhysRevA.105.052202}{Phys. Rev. A {\bf 105}, 052202}~(2022).

\bibitem{Mitra2023}
Arindam Mitra and M\'at\'e Farkas.
\newblock ``Characterizing and quantifying the incompatibility of quantum instruments''.
\newblock \href{https://dx.doi.org/10.1103/PhysRevA.107.032217}{Phys. Rev. A {\bf 107}, 032217}~(2023).

\bibitem{Leppaejaervi2024}
Leevi Lepp{\"{a}}j{\"{a}}rvi and Michal Sedl{\'{a}}k.
\newblock ``Incompatibility of quantum instruments''.
\newblock \href{https://dx.doi.org/10.22331/q-2024-02-12-1246}{{Quantum} {\bf 8}, 1246}~(2024).

\bibitem{buscemi2023unifying}
Francesco Buscemi, Kodai Kobayashi, Shintaro Minagawa, Paolo Perinotti, and Alessandro Tosini.
\newblock ``Unifying different notions of quantum incompatibility into a strict hierarchy of resource theories of communication''.
\newblock \href{https://dx.doi.org/10.22331/q-2023-06-07-1035}{{Quantum} {\bf 7}, 1035}~(2023).

\bibitem{Ji2024}
Kaiyuan Ji and Eric Chitambar.
\newblock ``Incompatibility as a resource for programmable quantum instruments''.
\newblock \href{https://dx.doi.org/10.1103/PRXQuantum.5.010340}{PRX Quantum {\bf 5}, 010340}~(2024).

\bibitem{Uola2020a}
Roope Uola, Tom Bullock, Tristan Kraft, Juha-Pekka Pellonp\"a\"a, and Nicolas Brunner.
\newblock ``All quantum resources provide an advantage in exclusion tasks''.
\newblock \href{https://dx.doi.org/10.1103/PhysRevLett.125.110402}{Phys. Rev. Lett. {\bf 125}, 110402}~(2020).

\bibitem{SudarsananRagini2024}
Nidhin Sudarsanan~Ragini and Sk~Sazim.
\newblock ``Higher-order incompatibility improves distinguishability of causal quantum networks''.
\newblock \href{https://dx.doi.org/10.1088/1367-2630/ad93f6}{New Journal of Physics {\bf 26}, 123003}~(2024).

\bibitem{jamiolkowski1972linear}
Andrzej Jamio{\l}kowski.
\newblock ``Linear transformations which preserve trace and positive semidefiniteness of operators''.
\newblock \href{https://dx.doi.org/https://doi.org/10.1016/0034-4877(72)90011-0}{Reports on Mathematical Physics {\bf 3}, 275--278}~(1972).

\bibitem{choi1975completely}
Man-Duen Choi.
\newblock ``Completely positive linear maps on complex matrices''.
\newblock \href{https://dx.doi.org/https://doi.org/10.1016/0024-3795(75)90075-0}{Linear algebra and its applications {\bf 10}, 285--290}~(1975).

\bibitem{slater2014lagrange}
Morton Slater.
\newblock ``Lagrange multipliers revisited''.
\newblock \href{https://dx.doi.org/10.1007/978-3-0348-0439-4_14}{Pages 293--306}.
\newblock Springer Basel. Basel~(2014).

\bibitem{kraus1983states}
Karl Kraus.
\newblock ``States, {E}ffects, and {O}perations: {F}undamental {N}otions of {Q}uantum {T}heory. lectures in {M}athematical {P}hysics at the {U}niversity of {T}exas at {A}ustin. lecture {N}otes in {P}hysics''.
\newblock \href{https://dx.doi.org/https://doi.org/10.1007/3-540-12732-1}{Volume 190}.
\newblock Springer, Berlin, Heidelberg. ~(1983).

\bibitem{schluck2023continuity}
Jonathan Schluck, Gláucia Murta, Hermann Kampermann, Dagmar Bruß, and Nikolai Wyderka.
\newblock ``Continuity of robustness measures in quantum resource theories''.
\newblock \href{https://dx.doi.org/10.1088/1751-8121/acd500}{Journal of Physics A: Mathematical and Theoretical {\bf 56}, 255303}~(2023).

\bibitem{skrzypczyk2023semidefinite}
Paul Skrzypczyk and Daniel Cavalcanti.
\newblock ``Semidefinite programming in quantum information science''.
\newblock \href{https://dx.doi.org/10.1088/978-0-7503-3343-6}{2053-2563}. IOP Publishing, Bristol, UK. ~(2023).

\bibitem{christandl2021asymptotic}
Matthias Christandl, Felix Leditzky, Christian Majenz, Graeme Smith, Florian Speelman, and Michael Walter.
\newblock ``Asymptotic performance of port-based teleportation''.
\newblock \href{https://dx.doi.org/10.1007/s00220-020-03884-0}{Communications in Mathematical Physics {\bf 381}, 379--451}~(2021).

\bibitem{sacchi2005optimal}
Massimiliano~F. Sacchi.
\newblock ``Optimal discrimination of quantum operations''.
\newblock \href{https://dx.doi.org/10.1103/PhysRevA.71.062340}{Phys. Rev. A {\bf 71}, 062340}~(2005).

\bibitem{saini2025completeness}
Rakesh Saini, Jukka Kiukas, Daniel Burgarth, and Alexei Gilchrist.
\newblock ``Completeness stability of quantum measurements''.
\newblock \href{https://dx.doi.org/10.1103/wkj6-l7bf}{Phys. Rev. A {\bf 113}, 012424}~(2026).

\end{thebibliography}
	
	\renewcommand{\theequation}{E\arabic{equation}}
	\setcounter{equation}{0}  

\onecolumn
\appendix

\section{Semidefinite programming (SDP) for incompatibility}\label{Appendix:SDP for incompatibility}
	\subsection{Generalized robustness of incompatibility (RoI) for a set of channels}
	While the dual formulation of RoI was introduced in Ref.~\cite{mori2020operational}, we detail its derivation for completeness, starting with the primal definition.
	The RoI of a set of channel is defined as follows:
	\begin{equation}
		\begin{split}
			R(\{\Lambda^{(i)}_{\rA\to\rB_i}\}_{i=1}^n):=&\min\Big\{r\ge 0~\big|~\cN_{\rA\to\rB_i}^{(i)}\in\sC(\rA\to\rB_i)~(i=1,\dots,n)\\
			&\qquad\mathrm{s.t.}\left\{\!\frac{\Lambda^{(i)}_{\rA\to\rB_i}\!+\!r\cN_{\rA\to\rB_i}^{(i)}}{1+r}\!\right\}_{i=1}^n\mbox{:~compatible}\Big\}.
		\end{split}
	\end{equation}
	By means of the Choi--Jamio\l kowski isomorphism~\cite{jamiolkowski1972linear,choi1975completely}, for a set of channels $\{\Lambda^{(i)}_{\rA\to\rB_i}\}_{i=1}^n$, the RoI $R(\{\Lambda^{(i)}_{\rA\to\rB_i}\}_{i=1}^n)$ is given by the solution of the following optimization problem with Choi operators.
	\begin{equation}
		\begin{split}
			R(\{\Lambda^{(i)}_{\rA\to\rB_i}\}_{i=1}^n)=&\underset{C_{\rB_i\rR}(\cN_{\rA\to\rB_i}^{(i)}),\;C_{\rB_1\cdots\rB_n\rR}}{\minimize} 
			\quad r\ge 0\\
			\subto \quad &\frac{C_{\rB_i\rR}(\Lambda_{\rA\to\rB_i})+rC_{\rB_i\rR}(\cN_{\rA\to\rB_i}^{(i)})}{1+r}=\Tr_{\rB_{j\neq i}} [C_{\rB_1\cdots\rB_n\rR}]\quad(i=1,\dots,n),\\
			&\Tr_{\rB_i}[C_{\rB_i\rR}(\cN_{\rA\to\rB_i}^{(i)})]=\one_{\rR},\\
            &C_{\rB_i\rR}(\cN_{\rA\to\rB_i}^{(i)})\succeq\zero_{\rB_i\rR},\\
			&C_{\rB_1\cdots\rB_n\rR}\succeq\zero_{\rB_1\cdots\rB_n\rR}.
		\end{split}
	\end{equation}
	Here $C_{\rB_i\rR}(\Lambda_{\rA\to\rB_i})$ denotes the Choi operator of the channel $\Lambda_{\rA\to\rB_i}$ and $\cH_\rR$ represents a reference Hilbert space isomorphic to $\cH_\rA$ (similarly for the other channels).
	Let us introduce
	\begin{align}
		\tilde{C}_{\rB_i\rR}(\cN_{\rA\to\rB_i}^{(i)}) &:= rC_{\rB_i\rR}(\cN_{\rA\to\rB_i}^{(i)}) ,\\
        \tilde{C}_{\rB_1\cdots\rB_n\rR}&:=(1+r)C_{\rB_1\cdots\rB_n\rR}.
	\end{align}
	Then the optimization becomes
	\begin{equation}
		\begin{split}
			R(\{\Lambda^{(i)}_{\rA\to\rB_i}\}_{i=1}^n)=&\underset{\tilde{C}_{\rB_i\rR}(\cN_{\rA\to\rB_i}^{(i)}), \tilde{C}_{\rB_1\cdots\rB_n\rR}}{\minimize}\quad r\ge 0\\
			\subto \quad &C_{\rB_i\rR}(\Lambda_{\rA\to\rB_i})+\tilde{C}_{\rB_i\rR}(\cN_{\rA\to\rB_i}^{(i)})=\Tr_{\rB_{j\neq i}} [\tilde{C}_{\rB_1\cdots\rB_n\rR}],\\
			&\Tr_{\rB_i} [\tilde{C}_{\rB_i\rR}(\cN_{\rA\to\rB_i}^{(i)})]=r\one_{\rR},\\
            &\tilde{C}_{\rB_i\rR}(\cN_{\rA\to\rB_i}^{(i)})\succeq\zero_{\rB_i\rR},\\
			&\tilde{C}_{\rB_1\cdots\rB_n\rR}\succeq\zero_{\rB_1\cdots\rB_n\rR}.
		\end{split}
	\end{equation}
	Note that the above constraints imply
	\begin{equation} \label{eq:simplify_SDP_conditions}
		\begin{split}
			\tilde{C}_{\rB_i\rR}(\cN_{\rA\to\rB_i}^{(i)})&=\Tr_{\rB_{j\neq i}} [\tilde{C}_{\rB_1\cdots\rB_n\rR}]-C_{\rB_i\rR}(\Lambda_{\rA\to\rB_i})\succeq\zero_{\rB_i\rR},\\
			\Tr_{\rB_1\cdots\rB_n}[\tilde{C}_{\rB_1\cdots\rB_n\rR}]&=\Tr_{\rB_i}[C_{\rB_i\rR}(\Lambda_{\rA\to\rB_i})]+\Tr_{\rB_i}[\tilde{C}_{\rB_i\rR}(\cN_{\rA\to\rB_i}^{(i)})]\\
			&=(1+r)\one_{\rR}.
		\end{split}    
	\end{equation}
	We thus rewrite the RoI of channels as follows:
	\begin{equation}\label{sdp:robustness_channel-channel}
		\begin{split}
			R(\{\Lambda^{(i)}_{\rA\to\rB_i}\}_{i=1}^n)=\underset{\tilde{C}_{\rB_1\cdots\rB_n\rR}}{\minimize}\quad &r\ge 0\\
			\subto \quad &\Tr_{\rB_1\cdots\rB_n}[\tilde{C}_{\rB_1\cdots\rB_n\rR}]=(1+r)\one_{\rR}\\
			&\Tr_{\rB_{j\neq i}} [\tilde{C}_{\rB_1\cdots\rB_n\rR}]-C_{\rB_i\rR}(\Lambda_{\rA\to\rB_i})\succeq\zero_{\rB_i\rR}\quad(i=1\dots,n),\\
			&\tilde{C}_{\rB_1\cdots\rB_n\rR}\succeq\zero_{\rB_1\cdots\rB_n\rR}.
		\end{split}
	\end{equation}
	Now, let us derive the dual problem.
	We introduce the Lagrangian by
	\begin{equation}
		\begin{split}
			L:=&~r-\Tr[V_\rR( (1+r)\one_\rR -\Tr_{\rB_1\cdots\rB_n}[\tilde{C}_{\rB_1\cdots\rB_n\rR}])] \\
            &-\sum_{i=1}^{n}\Tr[X_{\rB_i\rR}(\Tr_{\rB_{j\neq i}}[\tilde{C}_{\rB_1\cdots\rB_n\rR}]-C_{\rB_i\rR}(\Lambda_{\rA\to\rB_i}))] -\Tr[Z_{\rB_1\cdots\rB_n\rR}\tilde{C}_{\rB_1\cdots\rB_n\rR}],
		\end{split}
	\end{equation}
	where $V_\rR\in\cL_{\rH}(\cH_\rR)$, $X_{\rB_i\rR}\succeq\zero_{\rB_i\rR}$, $Z_{\rB_1\cdots\rB_n\rR}\succeq\zero_{\rB_1\cdots\rB_n\rR}$, and $r\ge 0$.
	Rewriting the Lagrangian leads to
	\begin{equation}
		\begin{split}
			L=&\sum_{i=1}^n\Tr[X_{\rB_i\rR}C_{\rB_i\rR}(\Lambda_{\rA\to\rB_i})]-\Tr[V_\rR] \\
			&-r(\Tr[V_\rR]-1)-\Tr\left[\left(\sum_{i=1}^n X_{\rB_i\rR}\otimes\one_{\rB_{j\neq i}} + Z_{\rB_1\cdots\rB_n\rR} - \one_{\rB_1\cdots\rB_n}\otimes V_\rR\right)\tilde{C}_{\rB_1\cdots\rB_n\rR}\right],
		\end{split}
	\end{equation}
    where $\one_{\rB_{j\neq i}}$ refers to the identity operator for all systems except the $i$-th system.
	Hence, the dual problem is 
	\begin{equation}\label{sdp:original_dual}
		\begin{split}
			\underset{X_{\rB_i\rR},V_\rR}{\maximize} \quad &\sum_{i=1}^n\Tr[X_{\rB_i\rR}C_{\rB_i\rR}(\Lambda_{\rA\to\rB_i})]-\Tr [V_\rR]\\
			\subto \quad & \Tr [V_\rR]\le 1,\\
			& \one_{\rB_1\cdots\rB_n}\otimes V_\rR \succeq \sum_{i=1}^n X_{\rB_i\rR}\otimes\one_{\rB_{j\neq i}}\\
			& X_{\rB_i\rR}\succeq \zero_{\rB_i\rR}\quad(i=1,\dots,n).
		\end{split}
	\end{equation}
	This SDP is strictly feasible.
	In fact, one can choose $V_\rR=\one_\rR/(d+1)$, $X_{\rB_i\rR}=\one_{\rB_i\rR}/2nd$ so that $\Tr[V_\rR]<1$, $\one_{\rB_1\cdots\rB_n}\otimes V_\rR \succ \sum_{i=1}^n X_{\rB_i\rR}\otimes\one_{\rB_{j\neq i}}$ and $X_{\rB_i\rR}\succ \zero_{\rB_i\rR}$.
	These guarantee the strong duality; the optimal value is achievable and both the primal and dual problems give the same result from Slater's condition~\cite{slater2014lagrange}.
	When $0<\Tr[V_\rR]<1$, with the renormalization $(X_{\rB_i\rR},V_\rR)\to\frac{1}{\Tr[V_\rR]}(X_{\rB_i\rR},V_\rR)$, the maximization is identified with the following SDP:
	\begin{equation}\label{sdp:RoI-dual}
		\begin{split}
			\underset{X_{\rB_i\rR},V_\rR}{\maximize} \quad &\sum_{i=1}^n\Tr[X_{\rB_i\rR}C_{\rB_i\rR}(\Lambda^{(i)}_{\rA\to\rB_i})]-1\\
			\subto \quad & \Tr [V_\rR]=1,\\
			& \one_{\rB_1\cdots\rB_n}\otimes V_\rR \succeq \sum_{i=1}^n X_{\rB_i\rR}\otimes\one_{\rB_{j\neq i}}\\
			& X_{\rB_i\rR}\succeq \zero_{\rB_i\rR}\quad(i=1,\dots,n).
		\end{split}
	\end{equation}

	When the SDP \eqref{sdp:original_dual} has an optimal solution with $\Tr[V_\rR]=0$ (hence $V_\rR=\zero_{\rR}$, $X_{\rB_i\rR}=\zero_{\rB_i\rR}$), the maximum is zero.	
	In this case, the maximum of the SDP~\eqref{sdp:RoI-dual} is also zero.
	To see this, we choose $V_\rR$ as a positive semidefinite and $\Tr [V_\rR]=1$, and set $X_{\rB_i\rR}=\one_{\rB_i}\otimes V_\rR$ for a certain $i$ and $X_{\rB_j\rR}=\zero_{\rB_j\rR}$ for all $j\neq i$.
	This assignment lets the objective value be zero, and since the maximum of the SDP~\eqref{sdp:RoI-dual} cannot exceed that of \eqref{sdp:original_dual}, it must be rigorously zero. 
    Therefore, hereinafter we use the expression \eqref{sdp:RoI-dual} as a dual problem of RoI.

	\subsection{Channel-POVM pair case}
	For a channel-POVM pair $(\Lambda_{\rA\to\rB},\E^\cX_\rA)$, we can derive an SDP to calculate its RoI simply by considering the measurement channel $\Gamma^{\E^\cX_\rA}$ of the POVM $\E^\cX_\rA$ and applying the SDP \eqref{sdp:robustness_channel-channel} to the channel pair $(\Lambda_{\rA\to\rB},\Gamma^{\E^\cX_\rA})$. 
	The formula $R(\Lambda_{\rA\to\rB},\E^\cX_\rA)=R(\Lambda_{\rA\to\rB},\Gamma^{\E^\cX_\rA})$~\cite{mori2020operational} ensures that the derived SDP is a valid protocol to calculate $R(\Lambda_{\rA\to\rB},\E^\cX_\rA)$.
	Here we present a self-contained derivation without using the result in Ref.~\cite{mori2020operational}.
	
	We start with a fundamental proposition.    
	\begin{proposition}\label{proposition:compatible}
		A channel-POVM pair $(\Lambda_{\rA\to\rB},\E^\cX_\rA)$ is compatible if and only if there is a family of linear operators $\{C^x_{\rB\rR}\}_{x\in\cX}$ such that
		\begin{align}
			\sum_{x\in\cX}C^x_{\rB\rR}&=C_{\rB\rR}(\Lambda_{\rA\to\rB}),\\
            \Tr_\rB C^x_{\rB\rR}&=E^{x\top}_\rA\quad(\forall x\in\cX),\\
            C^x_{\rB\rR}&\succeq\zero_{\rB\rR}\quad(\forall x\in\cX).
		\end{align}
	\end{proposition}
	\begin{proof}
		\textit{Only if part}.
		Take an $\E^\cX_\rA$-instrument $\{\cI^x_{\rA\to\rB}\}_{x\in\cX}$ such that $\Lambda_{\rA\to\rB}=\sum_{x\in\cX}\cI^x_{\rA\to\rB}$.
		Let $C_{\rB\rR}(\cI^x_{\rA\to\rB})$
		be the Choi operator of $\cI^x_{\rA\to\rB}$.
		The condition $\Lambda_{\rA\to\rB}=\sum_{x\in\cX}\cI^x_{\rA\to\rB}$ leads to $C_{\rB\rR}(\Lambda_{\rA\to\rB})=\sum_{x\in\cX}C_{\rB\rR}(\cI^x_{\rA\to\rB})$.
		To derive $\Tr_\rB C_{\rB\rR}(\cI^x_{\rA\to\rB})=E^{x\top}_\rA$, we assume that $\cI^x_{\rA\to\rB}$ has a Kraus representation~\cite{kraus1983states} $\cI^x_{\rA\to\rB}(\cdot)=\sum_i \kappa^{i|x}_{\rA\to\rB}(\cdot)\kappa^{i|x\dagger}_{\rA\to\rB}$.
		It must hold that
		\begin{equation}\label{eq:E-instrument_Kraus}
			\cI_{\rA\to\rB}^{x\dagger}(\one_\rB)=\sum_i \kappa^{i|x\dagger}_{\rA\to\rB}\kappa^{i|x}_{\rA\to\rB}=E^x_\rA.
		\end{equation}
		$\Tr_\rB C_{\rB\rR}(\cI^x_{\rA\to\rB})$ has the following form:
		\begin{equation}
			\begin{split}
				\Tr_\rB C_{\rB\rR}(\cI^x_{\rA\to\rB})&=\Tr_\rB[(\cI^x_{\rA\to\rB}\otimes\rmid_\rR)(\ketbra{\phi}{\phi}_{\rA\rR})]\\
				&=\sum_i\Tr_\rB[(\kappa^{i|x}_{\rA\to\rB}\otimes\one_\rR)(\ketbra{\phi}{\phi}_{\rA\rR})(\kappa^{i|x\dagger}_{\rA\to\rB}\otimes\one_\rR)],
			\end{split}
		\end{equation}
		where $\ket{\phi}_{\rA\rR}:=\sum_{i=0}^{d-1}\ket{i}_{\rA}\otimes\ket{i}_{\rR}$.
		Since we have
		\begin{equation}
			\begin{split}
				\Tr_\rB[(\kappa^{i|x}_{\rA\to\rB}\otimes\one_\rR)(\ketbra{\phi}{\phi}_{\rA\rR})(\kappa^{i|x\dagger}_{\rA\to\rB}\otimes\one_\rR)]&=\Tr_\rB\left[\sum_{kl}\kappa^{i|x}_{\rA\to\rB}\ketbra{k}{l}_\rA\kappa^{i|x\dagger}_{\rA\to\rB}\otimes\ketbra{k}{l}_\rR\right]\\
				&=\sum_{kl}\bra{l}_\rA\kappa^{i|x\dagger}_{\rA\to\rB}\kappa^{i|x}_{\rA\to\rB}\ket{k}_\rA\ketbra{k}{l}_\rR\\
				&=(\kappa^{i|x\dagger}_{\rA\to\rB}\kappa^{i|x}_{\rA\to\rB})^\top,
			\end{split}
		\end{equation}
		we obtain $\Tr_\rB C_{\rB\rR}(\cI^x_{\rA\to\rB})=\sum_i (\kappa^{i|x\dagger}_{\rA\to\rB}\kappa^{i|x}_{\rA\to\rB})^\top=E^{x\top}_\rA$, where the last equality comes from Eq.~\eqref{eq:E-instrument_Kraus}.
		
		\textit{If part}.
		Let us introduce a CP trace non-increasing map provided by Choi-Jamio\l kowski isomorphism~\cite{jamiolkowski1972linear,choi1975completely}
		\begin{equation}\label{eq:choi-jamiolkowski}
			\cJ^x_{\rA\to\rB}(\rho_\rA)=\Tr_\rR[(\one_\rB\otimes\rho^\top_\rR)C^x_{\rB\rR}].
		\end{equation}
		We now show that $\{\cJ^x_{\rA\to\rB}\}_{x\in\cX}$ is an $\E^x_\rA$-instrument such that $\sum_{x\in\cX}\cJ^x_{\rA\to\rB}=\Lambda_{\rA\to\rB}$.
		Taking the summation over all $x\in\cX$ of Eq.~\eqref{eq:choi-jamiolkowski} leads to
		\begin{equation}
			\begin{split}
				\sum_{x\in\cX}\cJ^x_{\rA\to\rB}(\rho_\rA)&=\sum_{x\in\cX}\Tr_\rR[(\one_\rB\otimes\rho^\top_\rR)C^x_{\rB\rR}]\\
				&=\Tr_\rR[(\one_\rB\otimes\rho^\top_\rR)C(\Lambda)_{\rB\rR}]\\
				&=\Lambda_{\rA\to\rB}(\rho_\rA).
			\end{split}
		\end{equation}
		We also have
		\begin{equation}
			\begin{split}
				\Tr \cJ^x_{\rA\to\rB}(\rho_\rA)&=\Tr[(\one_\rB\otimes\rho^\top_\rR)C^x_{\rB\rR}]\\
				&=\Tr[\rho^\top_\rA(\Tr_\rB C^x_{\rB\rR})]\\
				&=\Tr[\rho_\rA E^x_\rA]\quad(\forall \rho_\rA\in\cS(\cH_\rA)).
			\end{split}
		\end{equation}
		Hence $\Lambda_{\rA\to\rB}$ is realized by an $\E^\cX_\rA$-instrument, which completes the proof.
	\end{proof}    
	
	The RoI of a channel-POVM pair $(\Lambda_{\rA\to\rB},\E^\cX_\rA)$ is defined as follows:
	\begin{equation}
		\begin{split}
			R(\Lambda_{\rA\to\rB},\E^\cX_\rA):=\min\Big\{
			r\ge 0~\big|&~\cN_{\rA\to\rB}\in\sC(\rA\to\rB),\;\N^\cX_\rA\in\sM(\rA)\\
			&\mathrm{s.t.}\left(\frac{\Lambda_{\rA\to\rB}+r\cN_{\rA\to\rB}}{1+r},\frac{\E_{\rA}^\cX+r\N_{\rA}^{\cX}}{1+r}\right)\mbox{: compatible}\Big\}.
		\end{split}
	\end{equation} 
	Due to Proposition~\ref{proposition:compatible}, it can be rewritten as
	\begin{equation}
		\begin{split}
			R(\Lambda_{\rA\to\rB},\E^\cX_\rA)=\underset{C_{\rB\rR}(\cN_{\rA\to\rB}),N^\cX_\rA,C^x_{\rB\rR}}{\minimize} 
			\quad &r\ge 0\\
			\subto \quad &\frac{C_{\rB\rR}(\Lambda_{\rA\to\rB})+rC_{\rB\rR}(\cN_{\rA\to\rB})}{1+r}=\sum_{x\in\cX} C^x_{\rB\rR}\\
			&\frac{E^{x\top}_\rA + rN^{x\top}_\rA}{1+r} = \Tr_\rB[C^x_{\rB\rR}] \quad(\forall x\in\cX)\\
			&\Tr_\rB[C_{\rB\rR}(\cN_{\rA\to\rB})]=\one_{\rR},\quad C_{\rB\rR}(\cN_{\rA\to\rB})\succeq\zero_{\rB\rR}\\
			&N^x_\rA \succeq \zero_\rA \quad (\forall x\in\cX), \quad \sum_{x\in\cX} N^x_\rA =\one_\rA\\
			&C^x_{\rB\rR}\succeq\zero_{\rB\rR} \quad (\forall x\in\cX).
		\end{split}
	\end{equation}
    Here, $\rR$ has the same dimensional Hilbert space as $\rA$.
	Similarly to the channel-channel case, we can simplify the problem by defining
	\begin{equation}
	\tilde{C}_{\rB\rR}(\cN_{\rA\to\rB}) := rC_{\rB\rR}(\cN_{\rA\to\rB}),\quad
	\tilde{N}_\rA^{x} := r N_\rA^x, \quad
	\tilde{C}^x_{\rB\rR}:= (1+r)C^x_{\rB\rR}.
	\end{equation}
	Applying an argument analogous to Eq.~\eqref{eq:simplify_SDP_conditions} and using Proposition~\ref{proposition:compatible}, we can reduce the problem to 
	\begin{equation}\label{sdp:channel-POVM_robustness}
		\begin{split}
			R(\Lambda_{\rA\to\rB},\E^\cX_\rA)=\underset{\{\tilde{C}^x_{\rB\rR}\}}{\minimize} 
			\quad &r\ge 0\\
			\subto \quad &\sum_{x\in\cX} \Tr_\rB [\tilde{C}^x_{\rB\rR}] = (1+r)\one_\rR\\
			&\sum_{x\in\cX} \tilde{C}^x_{\rB\rR}-C_{\rB\rR}(\Lambda_{\rA\to\rB}) \succeq \zero_{\rB\rR}  \\
			&\Tr_\rB[\tilde{C}^x_{\rB\rR}]-E_\rA^{x\top} \succeq\zero_{\rA}\quad (\forall x\in\cX)\\
			&\tilde{C}^x_{\rB\rR}\succeq\zero_{\rB\rR} \quad (\forall x\in\cX).
		\end{split}
	\end{equation}
	We now derive the dual problem by constructing the following Lagrangian:
	\begin{equation}
		\begin{split}
			L := &~r-\Tr\Big[V_\rR \Big((1+r)\one_\rR - \sum_{x\in\cX} \Tr_\rB [\tilde{C}^x_{\rB\rR}]\Big)\Big] - \Tr\Big[X_{\rB\rR}\Big(\sum_{x\in\cX} \tilde{C}^x_{\rB\rR}-C_{\rB\rR}(\Lambda_{\rA\to\rB})\Big)\Big]\\
			&-\sum_{x\in\cX}\Tr\left[Y^x_\rR \left(\Tr_\rB[\tilde{C}^x_{\rB\rR}]-E_\rA^{x\top}\right)\right] - \sum_{x\in\cX}\Tr[Z^x_{\rB\rR}\tilde{C}^x_{\rB\rR}]\\
			= &\Tr[X_{\rB\rR}C_{\rB\rR}(\Lambda_{\rA\to\rB})] + \sum_{x\in\cX}\Tr\left[Y^x_\rR E_\rA^{x\top}\right] - \Tr[V_\rR] \\
			&- r (\Tr[V_\rR]-1) -\sum_{x\in\cX}\Tr[(X_{\rB\rR} +\one_\rB\otimes Y^x_\rR +Z^x_{\rB\rR} - \one_\rB\otimes V_\rR)\tilde{C}^x_{\rB\rR}],
		\end{split}
	\end{equation}
	where $V_\rR\in\cL_{\rH}(\cH_\rR)$, $X_{\rB\rR}\succeq\zero_{\rB\rR}$, $Y^x_\rR\succeq\zero_\rR$, $Z^x_{\rB\rR}\succeq\zero_{\rB\rR}\,(\forall x\in\cX)$, and $r\ge 0$.
	The dual problem is formulated as follows:
	\begin{equation}\label{eq:dual_robustness_channel-POVM}
		\begin{split}
			\underset{X_{\rB\rR},\{Y^x_\rR\},V_\rR}{\maximize} \quad &\Tr[X_{\rB\rR}C_{\rB\rR}(\Lambda_{\rA\to\rB})]+\sum_{x\in\cX}\Tr[Y^{x}_{\rR}E_\rA^{x\top}]-\Tr [V_\rR]\\
			\subto \quad & \Tr [V_\rR]\le1\\
			& \one_\rB \otimes V_\rR \succeq X_{\rB\rR} +\one_{\rB}\otimes Y^{x}_{\rR}\quad(\forall x\in\cX)\\
			& X_{\rB\rR}\succeq \zero_{\rB\rR}, \quad Y^x_{\rR}\succeq \zero_\rR \quad (\forall x\in\cX).
		\end{split}
	\end{equation}
	This problem is also strictly feasible; for instance, we can choose $V_\rR=\one_\rR/(d+1)$, $X_{\rB\rR}=\one_\rB\otimes\one_\rR/4d$, and $Y^x_\rR =\one_\rR/4d ~(\forall x\in\cX)$.
	Thus, from Slater's condition~\cite{slater2014lagrange}, the strong duality holds.
	The maximization is simplified as
	\begin{equation}\label{eq:dual_robustness_channel-POVM2}
		\begin{split}
			\underset{X_{\rB\rR},\{Y^x_\rR\},V_\rR}{\maximize} \quad &\Tr[X_{\rB\rR}C_{\rB\rR}(\Lambda_{\rA\to\rB})]+\sum_{x\in\cX}\Tr[Y^{x}_{\rR}E_\rA^{x\top}]-1\\
			\subto \quad & \Tr [V_\rR]=1\\
			& \one_\rB \otimes V_\rR \succeq X_{\rB\rR} +\one_{\rB}\otimes Y^{x}_{\rR}\quad(\forall x\in\cX)\\
			& X_{\rB\rR}\succeq \zero_{\rB\rR}, \quad Y^x_{\rR}\succeq \zero_\rR \quad (\forall x\in\cX).
		\end{split}
	\end{equation}
	We note that the optimization \eqref{eq:dual_robustness_channel-POVM2} is equal to the SDP \eqref{sdp:RoI-dual} for the channel pair $(\Lambda_{\rA\to\rB},\Gamma^{\E^\cX_\rA})$.
	In fact, setting $C_{\rX\rR}(\Gamma^{\E^\cX_\rA})=\sum_{x\in\cX} \ketbra{x}{x}\otimes E_\rA^{x\top}$ in $\eqref{sdp:RoI-dual}$ reproduces \eqref{eq:dual_robustness_channel-POVM2} and thus we can verify the formula $R(\Lambda_{\rA\to\rB},\E^\cX_\rA)=R(\Lambda_{\rA\to\rB},\Gamma^{\E^\cX_\rA})$ in Ref.~\cite{mori2020operational} also in our framework.
	
	\subsection{POVM-POVM pair case}
    As in the case of channel-POVM pairs, for a pair $(\E^\cX_\rA,\F^\cY_\rA)$ of POVMs, we can calculate its RoI by applying the SDP \eqref{sdp:RoI-dual} to the pair $(\Gamma^{\E^\cX_\rA}, \Gamma^{\F^\cY_\rA})$ of measurement channels.
    Yet, we again present a self-contained derivation without relying on this approach.

    The RoI $R(\E^\cX_\rA,\F^\cY_\rA)$ is given by the solution of the following optimization problem.
	\begin{equation}
		\begin{split}
			R(\E^\cX_\rA,\F^\cY_\rA)=\underset{\{N^x_\rA\},\{M^y_\rA\},\{G^{xy}_\rA\}}{\minimize} 
			\quad &r\ge 0\\
			\subto \quad &\frac{E^x_\rA + r N^x_\rA}{1+r} = \sum_{y\in\cY} G_\rA^{xy},\quad
			\frac{F^y_\rA + r M^y_\rA}{1+r} = \sum_{x\in\cX} G_\rA^{xy}\\
			&N^x_\rA \succeq \zero_\rA \quad (\forall x\in\cX), \quad \sum_{x\in\cX} N^x_\rA =\one_\rA\\
			&M^y_\rA \succeq \zero_\rA \quad (\forall y\in\cY), \quad \sum_{y\in\cY} M^y_\rA =\one_\rA\\
			&G_\rA^{xy}\succeq \zero_\rA \quad (\forall x\in\cX,\, \forall y\in\cY),\quad \sum_{x\in\cX,y\in\cY}G_\rA^{xy}=\one_\rA.
		\end{split}
	\end{equation}
    We recast the problem in a linear form using the variables:
	\begin{equation}
		\tilde{N}_\rA^{x} := r N_\rA^x, \quad \tilde{M}^y_\rA := rM^y_\rA, \quad \tilde{G}^{xy}_\rA:=(1+r)G^{xy}_\rA.
	\end{equation}
    A similar argument to Eq.~\eqref{eq:simplify_SDP_conditions} simplifies the problem to 
	\begin{equation}
		\begin{split}
			R(\E^\cX_\rA,\F^\cY_\rA)=\underset{\{\tilde{G}^{xy}_\rA\}}{\minimize} 
			\quad &r\ge 0\\
			\subto \quad &\sum_{x\in\cX,y\in\cY}\tilde{G}_\rA^{xy}=(1+r)\one_\rA\\
			&\sum_{y\in\cY} \tilde{G}_\rA^{xy} - E^x_\rA\succeq \zero_\rA \quad (\forall x\in\cX)\\
			&\sum_{x\in\cX} \tilde{G}_\rA^{xy} - F^y_\rA\succeq \zero_\rA \quad (\forall y\in\cY) \\
			&\tilde{G}_\rA^{xy}\succeq \zero_\rA \quad (\forall x\in\cX,\, \forall y\in\cY).
		\end{split}
	\end{equation}
	To derive the dual problem, we construct the Lagrangian as
	\begin{equation}
		\begin{split}
			L:=&~r - \Tr\Big[V_\rA \Big((1+r)\one_\rA - \sum_{x\in\cX,y\in\cY}\tilde{G}_\rA^{xy}\Big)\Big] - \sum_{x\in\cX}\Tr\Big[X^x_\rA \Big(\sum_{y\in\cY} \tilde{G}_\rA^{xy} - E^x_\rA\Big)\Big] \\
			&\qquad\qquad\qquad\qquad\qquad\qquad
			-\sum_{y\in\cY}\Tr\Big[Y^y_\rA \Big(\sum_{x\in\cX} \tilde{G}_\rA^{xy} - F^y_\rA\Big)\Big] - \sum_{x\in\cX,y\in\cY} \Tr[Z^{xy}_\rA \tilde{G}_\rA^{xy}]\\
			=& \sum_{x\in\cX}\Tr[X^x_\rA E^x_\rA] + \sum_{y\in\cY}\Tr[Y^y_\rA F^y_\rA] - \Tr[V_\rA]-r(\Tr[V_\rA]-1) \\
            &\qquad\qquad\qquad\qquad\qquad\qquad
            -\sum_{x\in\cX,y\in\cY} \Tr[(X^x_\rA + Y^y_\rA + Z^{xy}_\rA - V_\rA)\tilde{G}_\rA^{xy}],
		\end{split}    
	\end{equation}
	where $V_\rA\in\cL_{\rH}(\cH_\rA)$, $X^x_\rA\succeq\zero_\rA$, $Y^y_\rA\succeq\zero_\rA$, $Z^{xy}_\rA\succeq\zero_\rA \, (\forall x\in\cX,\, \forall y\in \cY)$, and $r\ge 0$.
	This yields the dual problem.
	\begin{equation}
		\begin{split}
			\underset{\{X^x_\rA\},\{Y^y_\rA\},V_\rA}{\maximize} \quad &\sum_{x\in\cX}\Tr[X^x_\rA E^x_\rA] + \sum_{y\in\cY}\Tr[Y^y_\rA F^y_\rA]-\Tr [V_\rA]\\
			\subto \quad & \Tr [V_\rA]\le1\\
			&  V_\rA  \succeq X^x_\rA + Y^y_\rA  \quad(\forall x\in\cX,\, \forall y\in\cY)\\
			& X^x_\rA\succeq \zero_\rA \quad(\forall x\in\cX), \quad Y^y_{\rA}\succeq \zero_\rA \quad (\forall y\in\cY).
		\end{split}
	\end{equation}  
	Again, strictly feasible solutions exist for this problem; for example, setting $V_\rA=\one_\rA/(d+1)$, $X^x_{\rA}=Y^y_\rA=\one_\rA/4d\,(\forall x\in\cX,\,\forall y\in\cY)$ strictly satisfies all inequalities. 
	The maximization is equivalent to 
	\begin{equation}
		\begin{split}
			\underset{\{X^x_\rA\},\{Y^y_\rA\},V_\rA}{\maximize} \quad &\sum_{x\in\cX}\Tr[X^x_\rA E^x_\rA] + \sum_{y\in\cY}\Tr[Y^y_\rA F^y_\rA]-1\\
			\subto \quad & \Tr [V_\rA]=1\\
			&  V_\rA  \succeq X^x_\rA + Y^y_\rA  \quad(\forall x\in\cX,\, \forall y\in\cY)\\
			& X^x_\rA\succeq \zero_\rA \quad(\forall x\in\cX), \quad Y^y_{\rA}\succeq \zero_\rA \quad (\forall y\in\cY).
		\end{split}
	\end{equation}  	

	\section{Proof of Theorem~\ref{thm:UR_GR}}\label{Appendix:Proof of Theorem1}
	For a linear map $\Psi_{\rA\to\rB}\colon\cL(\cH_\rA)\to\cL(\cH_\rB)$, its \emph{Choi operator} \cite{choi1975completely} is $C_{\rB\rR}(\Psi_{\rA\to\rB}):=(\Psi_{\rA\to\rB}\otimes\rmid_\rR)(\ketbra{\phi}{\phi}_{\rA\rR})$, where $\rR\simeq\rA$ is a reference system with $\dim\cH_\rR=\dim\cH_\rA=d$ and $\ket{\phi}_{\rA\rR}:=\sum_{i=0}^{d-1}\ket{i}_{\rA}\otimes\ket{i}_{\rR}$ is the unnormalized maximally entangled state.
	For a CPTP map $\Phi_{\rA\to\rB}$, it holds that $C_{\rB\rR}(\Phi_{\rA\to\rB})\succeq\zero_{\rB\rR}$ and $\Tr_\rB[C_{\rB\rR}(\Phi_{\rA\to\rB})]=\one_\rR$.

   	We establish a Lipschitz-type evaluation for the RoI of channels in terms of the diamond norm distance. 
   	We note that similar properties for robustness of states were investigated in Ref.~\cite{schluck2023continuity}.
	\begin{proposition}
		\label{prop:Lip}
		Let $\{\Lambda^{(i)}_{\rA\to\rB_i}\}_{i=1}^n,\{\Xi^{(i)}_{\rA\to\rB_i}\}_{i=1}^n$ be sets of quantum channels.
		The RoI of these channels satisfies 
		\begin{equation}\label{eq:Lip}
			2\big|R(\{\Lambda^{(i)}_{\rA\to\rB_i}\}_{i=1}^n)-R(\{\Xi^{(i)}_{\rA\to\rB_i}\}_{i=1}^n)\big|\le\sum_{i=1}^n\left\|\Lambda^{(i)}_{\rA\to\rB_i}-\Xi^{(i)}_{\rA\to\rB_i}\right\|_\diamond.
		\end{equation}
	\end{proposition}
	\begin{proof}
		We introduce a set $\sF$ by
		\begin{equation}
        \begin{split}
            \sF:=\{&(X_{\rB_1\rR},\dots,X_{\rB_n\rR},V_\rR)\mid\\
            &\Tr [V_\rR]=1,~\one_{\rB_1\cdots\rB_n}\otimes V_\rR \succeq \sum_{i=1}^n X_{\rB_i\rR}\otimes\one_{\rB_{j\neq i}},\;X_{\rB_i\rR}\succeq\zero_{\rB_i\rR}\}.
        \end{split}
		\end{equation}
		By using SDP~\eqref{sdp:RoI-dual}, we obtain
		\begin{widetext}
			\begin{align}
				&R(\{\Lambda^{(i)}_{\rA\to\rB_i}\}_{i=1}^n)-R(\{\Xi^{(i)}_{\rA\to\rB_i}\}_{i=1}^n)\\
                &=\max_{(X_{\rB_1\rR},\dots,X_{\rB_n\rR},V_\rR)\in\sF}\left\{\sum_{i=1}^n\Tr[X_{\rB_i\rR}C_{\rB_i\rR}(\Lambda^{(i)}_{\rA\to\rB_i})]-1\right\}\\
                &\qquad-\max_{(X_{\rB_1\rR},\dots,X_{\rB_n\rR},V_\rR)\in\sF}\left\{\sum_{i=1}^n \Tr[X_{\rB_i\rR}C_{\rB_i\rR}(\Xi^{(i)}_{\rA\to\rB_i})]-1\right\}\\
				&\le\max_{(X_{\rB_1\rR},\dots,X_{\rB_n\rR},V_\rR)\in\sF}\sum_{i=1}^n
				\left\{
				\Tr[X_{\rB_i\rR}C_{\rB_i\rR}(\Lambda^{(i)}_{\rA\to\rB_i}-\Xi^{(i)}_{\rA\to\rB_i})]\right\}\\
				&\le\sum_{i=1}^n
				\max_{(X_{\rB_i\rR},V_\rR)\in\sF_i}
				\left\{
				\Tr[X_{\rB_i\rR}C_{\rB_i\rR}(\Lambda^{(i)}_{\rA\to\rB_i}-\Xi^{(i)}_{\rA\to\rB_i})]\right\},\label{eq:prop_inequality2}
			\end{align}
            \end{widetext}
			where we introduced the sets
			\begin{equation}
				\sF_i:=\{(X_{\rB_i\rR},V_\rR)|\Tr [V_\rR]=1,~
                \zero_{\rB_i\rR}\preceq X_{\rB_i\rR}\preceq\one_{\rB_i}\otimes V_\rR\},
			\end{equation}
            which comes from tracing out $\rB_{j\neq i}$ from $\one_{\rB_1\cdots\rB_n}\otimes V_\rR \succeq \sum_{i=1}^n X_{\rB_i\rR}\otimes\one_{\rB_{j\neq i}}$.

			We next introduce another variable $\tilde{X}_{\rB_i\rR}:=2X_{\rB_i\rR}-\one_{\rB_i}\otimes V_\rR$ and we get
			\begin{align}
				&\max_{(X_{\rB_i\rR},V_\rR)\in\sF_i}
				\left\{
				\Tr[X_{\rB_i\rR}C_{\rB_i\rR}(\Lambda^{(i)}_{\rA\to\rB_i}-\Xi^{(i)}_{\rA\to\rB_i})]\right\}\\
				&=\frac{1}{2}\max_{(\tilde{X}_{\rB_i\rR},V_\rR)\in\tilde{\sF}_i}
				\left\{
				\Tr[\tilde{X}_{\rB_i\rR}C_{\rB_i\rR}(\Lambda^{(i)}_{\rA\to\rB_i}-\Xi^{(i)}_{\rA\to\rB_i})]\right.
				\\
                &\left.\qquad\qquad\qquad+
				\Tr[(\one_{\rB_i}\otimes V_\rR)C_{\rB_i\rR}(\Lambda^{(i)}_{\rA\to\rB_i}-\Xi^{(i)}_{\rA\to\rB_i})]
				\right\},
			\end{align}
			where
			\begin{align}
				\tilde{\sF}_i:=\{(\tilde{X}_{\rB_i\rR},V_\rR)|
				\Tr [V_\rR]=1,~
                -\one_{\rB_i}\otimes V_\rR\preceq \tilde{X}_{\rB_i\rR}\preceq\one_{\rB_i}\otimes V_\rR\}.
			\end{align}
			Because $\Lambda^{(i)}_{\rA\to\rB_i}$ and $\Xi^{(i)}_{\rA\to\rB_i}$ are TP maps, we have $\Tr[(\one_{\rB_i}\otimes V_\rR)C_{\rB_i\rR}(\Lambda^{(i)}_{\rA\to\rB_i}-\Xi^{(i)}_{\rA\to\rB_i})=\Tr_\rR\big[V_\rR\Tr_{\rB_i}[C_{\rB_i\rR}(\Lambda^{(i)}_{\rA\to\rB_i}-\Xi^{(i)}_{\rA\to\rB_i})]\big]=\Tr[V_\rR\one_\rR]-\Tr[V_\rR\one_\rR]=0$.
			It implies that
            \begin{align}
				&\max_{(X_{\rB_i\rR},V_\rR)\in\sF_i}
				\left\{
				\Tr[X_{\rB_i\rR}C_{\rB_i\rR}(\Lambda^{(i)}_{\rA\to\rB_i}-\Xi^{(i)}_{\rA\to\rB_i})]\right\}\\
                &=\frac{1}{2}\max_{(\tilde{X}_{\rB_i\rR},V_\rR)\in\tilde{\sF}_i}
				\left\{
				\Tr[\tilde{X}_{\rB_i\rR}C_{\rB_i\rR}(\Lambda^{(i)}_{\rA\to\rB_i}-\Xi^{(i)}_{\rA\to\rB_i})]\right\}.
            \end{align}
            Therefore, we can rewrite Eq.~\eqref{eq:prop_inequality2} as
			\begin{align}
				&R(\{\Lambda^{(i)}_{\rA\to\rB_i}\}_{i=1}^n)-R(\{\Xi^{(i)}_{\rA\to\rB_i}\}_{i=1}^n)\\
				&\le
				\frac{1}{2}\sum_{i=1}^n\max_{(\tilde{X}_{\rB_i\rR},V_\rR)\in\tilde{\sF}_i}
				\left\{
				\Tr[\tilde{X}_{\rB_i\rR}C_{\rB_i\rR}(\Lambda^{(i)}_{\rA\to\rB_i}-\Xi^{(i)}_{\rA\to\rB_i})]\right\}.
			\end{align}	
		We use the following SDP expression of the diamond norm distance for given two channels $\Pi_{\rA\to\rB},\Upsilon_{\rA\to\rB}\in\sC(\rA\to\rB)$ based on Ref.~\cite{skrzypczyk2023semidefinite}:
		\begin{equation}\label{eq:diamond_norm}
			\begin{split}
				\|\Pi_{\rA\to\rB}-\Upsilon_{\rA\to\rB}\|_\diamond
				=\underset{Z_{\rB\rR},V_\rR}\maximize \quad &\Tr[Z_{\rB\rR}C_{\rB\rR}(\Pi_{\rA\to\rB}-\Upsilon_{\rA\to\rB})]\\
				\subto\quad 
                &\Tr [V_\rR]=1,\\
                &-\one_\rB\otimes V_\rR\preceq Z_{\rB\rR}\preceq \one_\rB\otimes V_\rR.
			\end{split}
		\end{equation}
        In the end, we obtain
        \begin{equation}
            R(\{\Lambda^{(i)}_{\rA\to\rB_i}\}_{i=1}^n)-R(\{\Xi^{(i)}_{\rA\to\rB_i}\}_{i=1}^n)\le\frac{1}{2}\sum_{i=1}^n\left\|\Lambda^{(i)}_{\rA\to\rB_i}-\Xi^{(i)}_{\rA\to\rB_i}\right\|_\diamond.
        \end{equation}
        Applying the same argument with the roles of $\{\Lambda^{(i)}_{\rA\to\rB_i}\}_{i=1}^n$ and $\{\Xi^{(i)}_{\rA\to\rB_i}\}_{i=1}^n$ interchanged gives the same upper bound on $R(\{\Xi^{(i)}_{\rA\to\rB_i}\}_{i=1}^n)-R(\{\Lambda^{(i)}_{\rA\to\rB_i}\}_{i=1}^n)$.
	\end{proof}
	
	\begin{proof}[Proof of Theorem~\ref{thm:UR_GR}]
        Set $\{\Xi^{(i)}_{\rA\to\rB_i}\}_{i=1}^n$ in Eq.~\eqref{eq:Lip} as compatible set of channels.
		Since $R(\{\Xi^{(i)}_{\rA\to\rB_i}\}_{i=1}^n)=0$, we obtain the lower bound.

        As for the upper bound, we use the same argument of Ref.~\cite[Proposition 4]{Ghai2025}.
        By the definition of the RoI, there are channels $\{\cN^{(i)}_{\rA\to\rB_i}\}_{i=1}^n$ such that
        \begin{equation}
            \left\{\frac{\Lambda^{(1)}_{\rA\to\rB_1}+R(\{\Lambda^{(i)}_{\rA\to\rB_i}\}_{i=1}^n)\cN^{(1)}_{\rA\to\rB_1}}{1+R(\{\Lambda^{(i)}_{\rA\to\rB_i}\}_{i=1}^n)},\dots,\frac{\Lambda^{(n)}_{\rA\to\rB_n}+R(\{\Lambda^{(i)}_{\rA\to\rB_i}\}_{i=1}^n)\cN^{(n)}_{\rA\to\rB_n}}{1+R(\{\Lambda^{(i)}_{\rA\to\rB_i}\}_{i=1}^n)}\right\}
        \end{equation}
        are compatible.
        It holds that
        \begin{equation}
            \begin{split}
                \left\|\Lambda^{(j)}_{\rA\to\rB_j}-\frac{\Lambda^{(j)}_{\rA\to\rB_j}+R(\{\Lambda^{(i)}_{\rA\to\rB_i}\}_{i=1}^n)\cN^{(j)}_{\rA\to\rB_j}}{1+R(\{\Lambda^{(i)}_{\rA\to\rB_i}\}_{i=1}^n)}\right\|_\diamond&=\frac{R(\{\Lambda^{(i)}_{\rA\to\rB_i}\}_{i=1}^n)}{1+R(\{\Lambda^{(i)}_{\rA\to\rB_i}\}_{i=1}^n)}\|\Lambda^{(j)}_{\rA\to\rB_j}-\cN^{(j)}_{\rA\to\rB_j}\|_\diamond\\
                &\le \frac{2R(\{\Lambda^{(i)}_{\rA\to\rB_i}\}_{i=1}^n)}{1+R(\{\Lambda^{(i)}_{\rA\to\rB_i}\}_{i=1}^n)}.
            \end{split}
        \end{equation}
        The inequality used the diamond norm $\|\Lambda^{(i)}_{\rA\to\rB_i}-\cN^{(i)}_{\rA\to\rB_i}\|_\diamond$ is at most $2$.
        Therefore, there is a joint realization whose total error is
        \begin{equation}
            \sum_{j=1}^n\left\|\Lambda^{(j)}_{\rA\to\rB_j}-\frac{\Lambda^{(j)}_{\rA\to\rB_j}+R(\{\Lambda^{(i)}_{\rA\to\rB_i}\}_{i=1}^n)\cN^{(j)}_{\rA\to\rB_j}}{1+R(\{\Lambda^{(i)}_{\rA\to\rB_i}\}_{i=1}^n)}\right\|_\diamond\le \frac{2nR(\{\Lambda^{(i)}_{\rA\to\rB_i}\}_{i=1}^n)}{1+R(\{\Lambda^{(i)}_{\rA\to\rB_i}\}_{i=1}^n)}.
        \end{equation}
	\end{proof}

    \section{Proof of Proposition~\ref{proposition:id_n_robustness}}

    \begin{proof}[Proof of Proposition~\ref{proposition:id_n_robustness}]
    First, we show $R(\{\rmid^{(i)}_{\rA\to\rA_i}\}_{i=1}^{n})\le\frac{(n-1)(d-1)}{n+d-1}$.
    Let us introduce a channel $\cN_{\rA\to\rA_i}:=\alpha\rmid_{\rA\to\rA_i}+(1-\alpha)\cD_{\rA\to\rA_i}$ where $-\frac{1}{d^2-1}\le\alpha\le 1$ so that $\cN_{\rA\to\rA_i}$ becomes a CPTP map.
    To obtain an upper bound of the RoI of $n$ identity channels $R(\{\rmid_{\rA\to\rA_i}\}_{i=1}^{n})$, we solve the following identity:
    \begin{equation}
        \frac{\rmid_{\rA\to\rA_i}+r\cN_{\rA\to\rA_i}}{1+r}=\eta\rmid_{\rA\to\rA_i}+(1-\eta)\cD_{\rA\to\rA_i}.
    \end{equation}
    We get
    \begin{equation}
        r=\frac{1-\eta}{\eta-\alpha}.
    \end{equation}
    Therefore, $\alpha=-\frac{1}{d^2-1}$ minimizes $r$ and thus,
    \begin{align}
        r&=\frac{(n-1)(d-1)}{n+d-1}\\
        \cN_\rA&=-\frac{1}{d^2-1}\rmid_\rA+\frac{d^2}{d^2-1}\cD_\rA
    \end{align}
    is a feasible solution of $R(\{\rmid_{\rA\to\rA_i}\}_{i=1}^{n})$, which means that
    \begin{equation}
        R(\{\rmid_{\rA\to\rA_i}\}_{i=1}^{n})\le\frac{(n-1)(d-1)}{n+d-1}.
    \end{equation}

    Next, we show $R(\{\rmid_{\rA\to\rA_i}\}_{i=1}^{n})\ge\frac{(n-1)(d-1)}{n+d-1}$.
    According to Eq.~\eqref{sdp:RoI-dual}, $R(\{\rmid_{\rA\to\rA_i}\}_{i=1}^{n})$ is 
    \begin{equation}\label{eq:RoI_dual_identity}
		\begin{split}
			\underset{X_{\rA\rR},V_\rR}{\maximize} \quad &
            \sum_{i=1}^n\Tr[X_{\rA_i\rR}\ketbra{\phi}{\phi}_{\rA_i\rR}]-1\\
			\subto \quad & \Tr [V_\rR]=1,\\
			& \one_{\rA_1\cdots\rA_n}\otimes V_\rR \succeq \sum_{i=1}^n X_{\rA_i\rR}\otimes\one_{\rA_{j\neq i}}\\
			& X_{\rA_i\rR}\succeq \zero_{\rA_i\rR}\quad(i=1,\dots,n),
		\end{split}
	\end{equation}
    where $\one_{\rA_{j\neq i}}$ denotes the tensor product of identity operators except for $\rA_i$ and $\ket{\phi}_{\rA_i\rR}:=\sum_{i=0}^{d-1}\ket{i}_{\rA_i}\otimes\ket{i}_{\rR}$ is the unnormalized maximally entangled state.
    Let us confirm
    \begin{align}
        X_{\rA_i\rR}&=\frac{1}{d(n+d-1)}\ketbra{\phi}{\phi}_{\rA_i\rR},\label{eq:feasible_candidate1}\\
        V_\rR&=\frac{\one_{\rR}}{d},\label{eq:feasible_candidate2}
    \end{align}
    which makes the objective function of Eq.~\eqref{eq:RoI_dual_identity}
    \begin{equation}
        \frac{1}{d(n+d-1)}\sum_{i=1}^n|\braket{\phi|\phi}_{\rA_i\rR}|^2-1=\frac{1}{d(n+d-1)}nd^2-1=\frac{(n-1)(d-1)}{n+d-1},
    \end{equation}  
    are feasible solutions.
    It is enough to show
    \begin{equation}
        \sum_{i=1}^n\ketbra{\phi}{\phi}_{\rA_i\rR}\otimes\one_{\rA_{j\neq i}}\preceq (n+d-1)\one_{\rA_1\dots\rA_n\rR}.
    \end{equation}
    By introducing a normalized maximally entangled state $\ket{\Phi}_{\rA_i\rR}:=\frac{1}{\sqrt{d}}\ket{\phi}_{\rA_i\rR}$, the above inequality becomes
    \begin{equation}
        \sum_{i=1}^n\ketbra{\Phi}{\Phi}_{\rA_i\rR}\otimes\one_{\rA_{j\neq i}}\preceq \left(1+\frac{n-1}{d}\right)\one_{\rA_1\dots\rA_n\rR}.
    \end{equation}
    Define projective operators $Q^{(i)}_{\rA_1\cdots\rA_n\rR}:=\ketbra{\Phi}{\Phi}_{\rA_i\rR}\otimes\one_{\rA_{j\neq i}}$.
    Then, the above inequality
    \begin{equation}\label{eq:Q_ineqality}
        \sum_{i=1}^nQ^{(i)}_{\rA_1\cdots\rA_n\rR}\preceq\left(1+\frac{n-1}{d}\right)\one_{\rA_1\cdots\rA_n\rR}
    \end{equation}
    can be shown by showing the following bound on the operator norm\footnote{We found that Ref.~\cite[Remark 3.8]{christandl2021asymptotic} gives an equality for the operator norm of essentially the same operator. We provide a direct proof of the inequality used here.}   
    \begin{equation}\label{eq:Q_op-norm}
        \left\|\left(\sum_{i=1}^n Q^{(i)}_{\rA_1\cdots\rA_n\rR}\right)\right\|\le 1+\frac{n-1}{d}.
    \end{equation}
    To evaluate the maximum eigenvalue of $\sum_{i=1}^n Q^{(i)}_{\rA_1\cdots\rA_n\rR}$, we introduce an isometry $W_i:\bigotimes_{j\neq i}\cH_{\rA_j}\to\cH_{\rA_1\cdots\rA_n\rR}$ such that
    \begin{equation}
        W_i\ket{\psi}_{\rA_{j\neq i}}=\ket{\Phi}_{\rA_i\rR}\otimes \ket{\psi}_{\rA_{j\neq i}}.
    \end{equation}
    We have
    \begin{equation}
        Q^{(i)}_{\rA_1\cdots\rA_n\rR}=W_iW^\dagger_i.
    \end{equation}
    Introduce
    \begin{equation}
        W:=(W_1,\dots,W_n)
    \end{equation}
    and we have
    \begin{equation}
        \sum_{i=1}^nQ^{(i)}_{\rA_1\cdots\rA_n\rR}=WW^\dagger.
    \end{equation}
    Notice that $\|WW^\dagger\|=\|W^\dagger W\|$.
    Then, we compute the maximum eigenvalue of $W^\dagger W$, that is,
    \begin{equation}
    W^\dagger W=
        \begin{pmatrix}
            W^\dagger_1W_1 & W^\dagger_1 W_2 & \cdots & W^\dagger_1W_n\\
            W^\dagger_2W_1 & W^\dagger_2W_2 & \cdots  & W^\dagger_2W_n\\
            \vdots & \vdots & \ddots & \vdots\\
            W^\dagger_nW_1 & W^\dagger_nW_2 & \cdots  & W^\dagger_n W_n
        \end{pmatrix}.
    \end{equation}
    Since $W_i$ is an isometry, $W^\dagger_i W_i=\one$.
    As for $W^\dagger_i W_j:\bigotimes_{k\neq j}\cH_{\rA_k}\to\bigotimes_{k\neq i}\cH_{\rA_k}$, take a basis vector $\ket{a}_{\rA_i}\otimes\bigotimes_{k\neq i,j}\ket{x_k}$ and we have
    \begin{equation}
        \begin{split}
            &W^\dagger_i W_j\left(\ket{a}_{\rA_i}\otimes\bigotimes_{k\neq i,j}\ket{x_k}\right)\\
            &=\left(\frac{1}{\sqrt{d}}\sum_{s=0}^{d-1}\bra{s}_{\rA_i}\otimes\one_{\rA_{k\neq i} }\otimes\bra{s}_\rR\right)\left(\frac{1}{\sqrt{d}}\sum_{r=0}^{d-1}\ket{a}_{\rA_i}\otimes\ket{r}_{\rA_j}\otimes\bigotimes_{k\neq i, j}\ket{x_k}\otimes\ket{r}_\rR\right)\\
            &=\frac{1}{d}\sum_{s,r=0}^{d-1}\braket{s|a}_{\rA_i}\braket{s|r}_\rR\ket{r}_{\rA_j}\otimes\bigotimes_{k\neq i,j}\ket{x_k}\\
            &=\frac{1}{d}\ket{a}_{\rA_j}\otimes\bigotimes_{k\neq i,j}\ket{x_k}.
        \end{split}
    \end{equation}
    Therefore, $W^\dagger_i W_j$ can be regarded as $\frac{1}{d}U_{ij}$ where $U_{ij}$ is a unitary that shifts a vector $\ket{\psi}_{\rA_i}$ to $\ket{\psi}_{\rA_j}$.
    Now, we rewrite $W^\dagger W$
    \begin{equation}
    W^\dagger W=
        \begin{pmatrix}
            \one_{\rA_{k\neq 1}} & \frac{1}{d}U_{12} & \cdots & \frac{1}{d}U_{1n}\\
            \frac{1}{d}U_{21} &  \one_{\rA_{k\neq 2}} & \cdots  & \frac{1}{d}U_{2n}\\
            \vdots & \vdots & \ddots & \vdots\\
            \frac{1}{d}U_{n1} & \frac{1}{d}U_{n2}& \cdots  &  \one_{\rA_{k\neq n}}
        \end{pmatrix}.
    \end{equation}
    For any $\ket{x}:=(\ket{x_1},\dots,\ket{x_n})^\top$ where $\ket{x_i}\in\bigotimes_{k\neq i}\cH_{\rA_k}$, one has
    \begin{equation}
        \begin{split}
            \bra{x}W^\dagger W\ket{x}&=\sum_{i=1}^n\|x_i\|^2+\frac{1}{d}\sum_{i\neq j}\bra{x_i}U_{ij}\ket{x_j}\\
            &\le \sum_{i=1}^n\|x_i\|^2+\frac{1}{d}\sum_{i\neq j}\|x_i\|\|U_{ij}x_j\|\\
            &=\sum_{i=1}^n\|x_i\|^2+\frac{1}{d}\sum_{i\neq j}\|x_i\|\|x_j\|.
        \end{split}
    \end{equation}
    Let us denote $\|x_i\|=s_i$ and consider the case where $\|x\|=\sqrt{\sum_{i=1}^ns^2_i}=1$.
    In this case,
    \begin{equation}
        \begin{split}
            \sum_{i=1}^n\|x_i\|^2+\frac{1}{d}\sum_{i\neq j}\|x_i\|\|x_j\|&=1+\frac{1}{d}\left\{\left(\sum_{i=1}^ns_i\right)^2-\sum_{i=1}^n s_i^2\right\}\\
            &\le 1+\frac{1}{d}\left(n\sum_{i=1}^ns^2_i-\sum_{i=1}^ns^2_i\right)\\
            &=1+\frac{n-1}{d}.
        \end{split}
    \end{equation}
    This shows Eq.~\eqref{eq:Q_op-norm}.

    Therefore, Eqs.~\eqref{eq:feasible_candidate1} and \eqref{eq:feasible_candidate2} are feasible solutions of Eq.~\eqref{eq:RoI_dual_identity}, which means $R(\{\rmid_{\rA\to\rA_i}\}_{i=1}^{n})\ge\frac{(n-1)(d-1)}{n+d-1}$.
    \end{proof}

    \section{Proof of Proposition~\ref{proposition:minimum_total_error_id}}\label{appendix:minimum_total_error_id}

    \begin{proof}[Proof of Proposition~\ref{proposition:minimum_total_error_id}]
    First, we show $E(\{\rmid^{(i)}_{\rA\to\rA_i}\}_{i=1}^{n})\ge \frac{2(n-1)(d-1)}{d}$.
    For any joint realization $\Theta_{\rA\to\rA_1\cdots\rA_n}$ and its marginal $\Theta^{(i)}_{\rA\to\rA_i}:=\Tr_{j\neq i}\Theta_{\rA\to\rA_1\cdots\rA_n}$, it holds that
    \begin{equation}
        \begin{split}
            \|\rmid^{(i)}_{\rA\to\rA_i}-\Theta^{(i)}_{\rA\to\rA_i}\|_\diamond&\ge \|\rmid^{(i)}_{\rA\to\rA_i}\otimes\rmid_\rR(\ketbra{\Phi}{\Phi}_{\rA\rR})-\Theta^{(i)}_{\rA\to\rA_i}\otimes\rmid_{\rR}(\ketbra{\Phi}{\Phi}_{\rA\rR})\|_1\\
            &=\|\ketbra{\Phi}{\Phi}_{\rA_i\rR}-\Theta^{(i)}_{\rA\to\rA_i}\otimes\rmid_{\rR}(\ketbra{\Phi}{\Phi}_{\rA\rR})\|_1.
        \end{split}
    \end{equation}
    Now, define $\omega_{\rA_i\rR}:=\Theta^{(i)}_{\rA\to\rA_i}\otimes\rmid_{\rR}(\ketbra{\Phi}{\Phi}_{\rA\rR})$ and consider a measurement that consists of projectors $\{\ketbra{\Phi}{\Phi}_{\rA_i\rR}, \one_{\rA_i\rR}-\ketbra{\Phi}{\Phi}_{\rA_i\rR}\}$.
    When one performs this measurement to a state $\ketbra{\Phi}{\Phi}_{\rA_i\rR}$, one gets a probability distribution $\bm{p}=(1,0)$.
    When one performs this measurement to a state $\omega_{\rA_i\rR}$, one gets
    \begin{equation}
        \bm{q}:=(\Tr[\ketbra{\Phi}{\Phi}_{\rA_i\rR}\omega_{\rA_i\rR}], 1-\Tr[\ketbra{\Phi}{\Phi}_{\rA_i\rR}\omega_{\rA_i\rR}]).
    \end{equation}
    From the monotonicity of the trace norm, we get
    \begin{equation}
    \begin{split}
        \|\rmid^{(i)}_{\rA\to\rA_i}-\Theta^{(i)}_{\rA\to\rA_i}\|_\diamond&\ge\|\bm{p}-\bm{q}\|_1=2(1-\Tr[\ketbra{\Phi}{\Phi}_{\rA_i\rR}\omega_{\rA_i\rR}]).
    \end{split}
    \end{equation}

    Define $\omega_{\rA_1\cdots\rA_n\rR}:=\Theta_{\rA\to\rA_1\cdots\rA_n}\otimes\rmid_\rR(\ketbra{\Phi}{\Phi}_{\rA\rR})$
    and a projection operator $Q^{(i)}_{\rA_1\cdots\rA_n\rR}:=\ketbra{\Phi}{\Phi}_{\rA_i\rR}\otimes\one_{j\neq i}$.
    Then since $\Tr[\ketbra{\Phi}{\Phi}_{\rA_i\rR}\omega_{\rA_i\rR}]=\Tr[Q^{(i)}_{\rA_1\cdots\rA_n\rR}\omega_{\rA_1\cdots\rA_n\rR}]$, one gets
    \begin{equation}
        \|\rmid^{(i)}_{\rA\to\rA_i}-\Theta^{(i)}_{\rA\to\rA_i}\|_\diamond\ge 2(1-\Tr[Q^{(i)}_{\rA_1\cdots\rA_n\rR}\omega_{\rA_1\cdots\rA_n\rR}]).
    \end{equation}
    Therefore, it holds that
    \begin{equation}
        \begin{split}
            \sum_{i=1}^n\|\rmid^{(i)}_{\rA\to\rA_i}-\Theta^{(i)}_{\rA\to\rA_i}\|_\diamond&\ge2\left(n-\Tr\left[\left(\sum_{i=1}^n Q^{(i)}_{\rA_1\cdots\rA_n\rR}\right)\omega_{\rA_1\cdots\rA_n\rR}\right]\right)\\
            &\ge 2\left(n-\Tr\left[\left(1+\frac{n-1}{d}\right)\omega_{\rA_1\cdots\rA_n\rR}\right]\right)\\
            &=2\left(n-\left(1+\frac{n-1}{d}\right)\right)\\
            &=\frac{2(n-1)(d-1)}{d}.
        \end{split}
    \end{equation}
    The second inequality uses Eq.~\eqref{eq:Q_ineqality}.
    
    Next, let us consider the total error of the optimal cloner.
    Notice that we have
        \begin{equation}
            \begin{split}
                \|\rmid_\rA-\cC_\rA\|_\diamond&=(1-\eta)\|\rmid_\rA-\cD_\rA\|_\diamond\\
                &=\frac{d(n-1)}{n(d+1)}\cdot2\left(\frac{d^2-1}{d^2}\right)\\
                &=\frac{2(n-1)(d-1)}{nd}.
            \end{split}
        \end{equation}
        The second equality uses the result of Ref.~\cite{sacchi2005optimal} stating that the error probability of discriminating the identity channel and the completely depolarizing channel is $\frac{1}{2d^2}$ and it is equal to $\frac{1}{2}-\frac{1}{4}\|\rmid_\rA-\cD_\rA\|_\diamond$~(see Refs.~\cite{wilde2017quantum,watrous2018theory}).
        Therefore, the total error of the optimal cloner is $\frac{2(n-1)(d-1)}{d}$.
    \end{proof}

    \section{Joint-realization tradeoffs from channel incompatibility}
    \label{Appendix:Joint-realization tradeoffs from channel incompatibility}

    \subsection{Proof of Theorem~\ref{theorem:mur}}
    \begin{proof}[Proof of Theorem~\ref{theorem:mur}]
        Let $\G^{(1)\cX}_\rA$ and $\G^{(2)\cY}_\rA$ be compatible POVMs.
        By using Theorem~\ref{thm:UR_GR} and Proposition~\ref{proposition:channel-measurement}, we get
        \begin{equation}\label{eq:2R(POVM)_diamond}
			\begin{split}
				2R\left(\E^\cX_\rA,\F^\cY_\rA\right)\le \left\|\Gamma^{\E^\cX_\rA}-\Gamma^{\G^{(1)\cX}_\rA}\right\|_\diamond+\left\|\Gamma^{\F^\cY_\rA}-\Gamma^{\G^{(2)\cY}_\rA}\right\|_\diamond.
			\end{split}
		\end{equation}

        Next, we evaluate these diamond norm-based terms by using the operator norm.
        For general POVMs $\E^\cX_\rA$ and $\G^\cX_\rA$, define $D^x_\rA:=E^x_\rA-G^x_\rA$.
        We have the following: 
		\begin{equation}
            \begin{split}
                &\|\Gamma^{\E^\cX_\rA}-\Gamma^{\G^\cX_\rA}\|_\diamond \\
                &=\sup_{\rho_{\rA\rR}\in\cS(\cH_\rA\otimes\cH_\rR)} \left\|\sum_{x\in\cX}\Tr_\rA[(E^x_\rA\otimes\one_\rR)\rho_{\rA\rR}]\otimes\ketbra{x}{x}_\rX-\sum_{x'\in\cX}\Tr_\rA[(G^{x'}_\rA\otimes\one_\rR)\rho_{\rA\rR}]\otimes\ketbra{x'}{x'}_\rX\right\|_1\\
                &=\sup_{\rho_{\rA\rR}\in\cS(\cH_\rA\otimes\cH_\rR)} \left\|\sum_{x\in\cX}\Tr_\rA[(D^x_\rA\otimes\one_\rR)\rho_{\rA\rR}]\otimes\ketbra{x}{x}_\rX\right\|_1\\
                &=\sup_{\rho_{\rA\rR}\in\cS(\cH_\rA\otimes\cH_\rR)} \sum_{x\in\cX}\!\left\| \Tr_\rA[(D^x_\rA\otimes \one_{\rR})\rho_{\rA\rR}] \right\|_1.
            \end{split}
		\end{equation}
        By using the relationship between trace norm and operator norm (see e.g.,\cite[Section 1.1.3]{watrous2018theory}), we can write
        \begin{equation}
            \begin{split}
                \left\| \Tr_\rA[(D^x_\rA\otimes \one_{\rR})\rho_{\rA\rR}] \right\|_1&=\max_{\|V_\rR\|\le 1}|\Tr\{V_\rR\Tr_\rA[(D^x_\rA\otimes \one_{\rR})\rho_{\rA\rR}]\}|\\
                &=\max_{\|V_\rR\|\le 1}|\Tr[(D^x_\rA\otimes V_\rR)\rho_{\rA\rR}]|.
            \end{split}
        \end{equation}
        For any state $\rho_{\rA\rR}\in\cS(\cH_\rA\otimes\cH_\rR)$ and operator $V_\rR\in\cL(\cH_\rR)$ such that $\|V_\rR\|\le 1$, it holds that
        \begin{equation}
            \begin{split}
                |\Tr[(D^x_\rA\otimes V_\rR)\rho_{\rA\rR}]|&\le\|D^x_\rA\otimes V_\rR\|\|\rho_{\rA\rR}\|_1\\
                &=\|D^x_\rA\otimes V_\rR\|\\
                &=\|D^x_\rA\|\|V_\rR\|\le\|D^x_\rA\|.
            \end{split}
        \end{equation}
        The inequality is H\"{o}lder's inequality.
        Hence we have
        \begin{equation}
            \|\Gamma^{\E^\cX_\rA}-\Gamma^{\G^\cX_\rA}\|_\diamond \le\sum_{x\in\cX}\|D^x_\rA\|=\sum_{x\in\cX}\|E^x_\rA-G^x_\rA\|.
        \end{equation}
        Combining this inequality and Eq.~\eqref{eq:2R(POVM)_diamond} completes the proof.
    \end{proof}
    \section{Information-error-disturbance tradeoff}
    \label{Appendix:information-error-disturbance tradeoff}
	\subsection{Proof of Proposition~\ref{proposition:RoM_vs_RoI}}
	\begin{proof}[Proof of Proposition~\ref{proposition:RoM_vs_RoI}]
		The proof proceeds by choosing a certain dual feasible solution of $R(\rmid_\rA,\E^\cX_\rA)$ that attains $\left(\sqrt{I(\E^\cX_\rA)+1}-1\right)^2/(d-1)$.
		From SDP~\eqref{eq:dual_robustness_channel-POVM2}, the dual problem of $R(\rmid_\rA,\E^\cX_\rA)$ is 
		\begin{equation}
			\begin{split}
				\underset{X_{\rA\rR},\{Y_{\rR}^x\},V_\rR}{\maximize} \quad &\Tr[X_{\rA\rR}\ketbra{\phi}{\phi}_{\rA\rR}]+\sum_{x\in\cX}\Tr[Y^{x}_\rR E^{x\top}_\rA]-1\\
				\subto \quad & \Tr [V_\rR]=1,\;V_\rR\in\cL_{\rH}(\cH_\rR)\\
				& \one_{\rA}\otimes V_\rR \succeq X_{\rA\rR} +\one_{\rA}\otimes Y_{\rR}^{x} \;(\forall x\in\cX)\\
				& X_{\rA\rR}\succeq \zero_{\rA\rR},\;Y_{\rR}^{x}\succeq \zero_{\rR}\;(\forall x\in\cX).
			\end{split}
		\end{equation}
		Here $\ket{\phi}_{\rA\rR}=\sum_{i=0}^{d-1}\ket{i}_\rA\otimes\ket{i}_\rR$ is the unnormalized maximally entangled state where $\{\ket{i}_\rA\}_{i=0}^{d-1}$ and $\{\ket{i}_\rR\}_{i=0}^{d-1}$ are standard basis of $\cH_\rA$ and $\cH_\rR$, respectively.
        Set $V_\rR=\one_\rR/d,X_{\rA\rR}=\alpha\ketbra{\phi}{\phi}_{\rA\rR}, Y^x_\rR=\beta\ketbra{\psi^*_x}{\psi^*_x}_\rR$ where $\alpha,\beta\ge 0$ and $\ket{\psi^*_x}_\rR$ is the eigenvector corresponding to the maximal eigenvalue of $E^{x\top}_\rA$.
		The constraint is then
		\begin{equation}
			\frac{1}{d}\one_\rA\otimes\one_\rR-\alpha\ketbra{\phi}{\phi}_{\rA\rR}-\beta\one_\rA\otimes\ketbra{\psi^*_x}{\psi^*_x}_\rR\succeq\zero_{\rA\rR}.
		\end{equation}
		Let $U^x_\rR$ be a unitary operator on $\cH_\rR$ such that $U^x_\rR\ket{\psi^*_x}_\rR=\ket{0}_\rR$.
        Let $U^x_\rA$ denote a unitary on $\cH_\rA$ represented by the same matrix as $U^x_\rR$ in the chosen basis.
        Let $U^{x*}_\rA$ be the entry-wise complex conjugation of $U^x_\rA$ with respect to the standard basis.
		Since it holds that $U^{x*}_\rA\otimes U^x_\rR\ket{\phi}_{\rA\rR}=\ket{\phi}_{\rA\rR}$, the constraint is equivalent to
		\begin{equation}
			M_{\rA\rR}:=\frac{1}{d}\one_\rA\otimes\one_\rR-\alpha\ketbra{\phi}{\phi}_{\rA\rR}-\beta\one_\rA\otimes\ketbra{0}{0}_\rR\succeq\zero_{\rA\rR}.
		\end{equation}
		For the standard basis $\{\ket{i}_\rA\}_{i=0}^{d-1}$ and $\{\ket{j}_\rR\}_{j=0}^{d-1}$ of $\cH_\rA$ and $\cH_\rR$, respectively, we have
		\begin{equation}
			\bra{i,j}M_{\rA\rR}\ket{k,l}_{\rA\rR}=\frac{1}{d}\delta_{i,k}\delta_{j,l}-\alpha\delta_{i,j}\delta_{k,l}-\beta\delta_{i,k}\delta_{j,0}\delta_{l,0}.
		\end{equation}
		Let us first focus on a subspace whose basis is$\{\ket{i,i}_{\rA\rR}\}_{i=0}^{d-1}$.
		Then
		\begin{equation}
			\bra{i,i}M_{\rA\rR}\ket{j,j}_{\rA\rR}=\frac{1}{d}\delta_{i,j}-\alpha-\beta\delta_{i,j}\delta_{i,0}.
		\end{equation}
		Therefore, one can introduce a matrix on $\mbC^d$
		\begin{equation}
			M^\cD_{\mbC^d}:=\frac{1}{d}\one_{\mathbb C^d}-\alpha J_d-\beta\ketbra{0}{0}
		\end{equation}
		where $J_d$ is a matrix whose elements are all $1$.
		Next, consider a subspace whose basis is $\{\ket{i,0}\}_{i=1}^{d-1}$.
		Then 
		\begin{equation}
			\bra{i,0}M_{\rA\rR}\ket{j,0}_{\rA\rR}=\frac{1}{d}\delta_{i,j}-\beta\delta_{i,j}.
		\end{equation}
		Finally, consider a subspace whose basis is $\{\ket{i,j}_{\rA\rR}\}_{\substack{i=0,j=1\\(i\neq j)}}^{d-1}$.
		Then
		\begin{equation}
			\bra{i,j}M_{\rA\rR}\ket{k,l}_{\rA\rR}=\frac{1}{d}\delta_{i,k}\delta_{j,l}.
		\end{equation}
		Therefore, $M_{\rA\rR}$ has the following decomposition:
		\begin{equation}
			M_{\rA\rR}= M^\cD_{\mbC^d}\oplus\left(\frac{1}{d}-\beta\right)\one_{\mbC^{d-1}}\oplus\frac{1}{d}\one_{\mathbb C^{(d-1)^2}}.
		\end{equation}
		
		Now, let us consider the condition for $M_{\rA\rR}\succeq\zero_{\rA\rR}$, which is equivalent to $M^\cD_{\mbC^d}\succeq\zero_{\mbC^d}$ and $\left(\frac{1}{d}-\beta\right)\one_{\mbC^{d-1}}\succeq\zero_{\mbC^{d-1}}$.
		The latter condition is clearly equivalent to $\frac{1}{d}\ge\beta$, so we will consider $M^\cD_{\mbC^d}\succeq\zero_{\mbC^d}$.
		Let us introduce a basis $\left\{\ket{0}_{\mbC^d},\frac{1}{\sqrt{d-1}}\sum_{i=1}^{d-1}\ket{i}_{\mbC^d}\right\}$.
		The orthogonal complement of $\mathrm{span}\left\{\ket{0}_{\mbC^d},\frac{1}{\sqrt{d-1}}\sum_{i=1}^{d-1}\ket{i}_{\mbC^d}\right\}$ consists of vectors such as $\ket{w}_{\mbC^d}$ where $\braket{0|w}_{\mbC^d}=0$ and $\frac{1}{\sqrt{d-1}}\sum_{i=1}^{d-1}\braket{i|w}_{\mbC^d}=0$.
		Note that for any non-zero vector $\ket{\varphi}\in\mathrm{span}\left\{\ket{0}_{\mbC^d},\frac{1}{\sqrt{d-1}}\sum_{i=1}^{d-1}\ket{i}_{\mbC^d}\right\}^\perp$, it holds that
		\begin{equation}
			\bra{\varphi}M^\cD_{\mbC^d}\ket{\varphi}=\frac{1}{d}\braket{\varphi|\varphi}-\alpha\bra{\varphi}J_d\ket{\varphi}-\beta\braket{\varphi|0}\braket{0|\varphi}=\frac{1}{d}\braket{\varphi|\varphi}>0.
		\end{equation}
		Thus, it is enough to consider only the $2$-dimensional subspace $\mathrm{span}\left\{\ket{0}_{\mbC^d},\frac{1}{\sqrt{d-1}}\sum_{i=1}^{d-1}\ket{i}_{\mbC^d}\right\}$.
		On this subspace, $M^\cD_{\mbC^d}$ reduces to a $2\times 2$ matrix
		\begin{equation}
			K:=\begin{pmatrix}
				\frac{1}{d}-\alpha-\beta & -\alpha\sqrt{d-1}\\
				-\alpha\sqrt{d-1} & \frac{1}{d}-\alpha(d-1)
			\end{pmatrix}.
		\end{equation}
		$K\succeq\zero$ if and only if
		\begin{align}
			\frac{1}{d}-\alpha-\beta &\ge 0,\\
			\frac{1}{d}-\alpha(d-1)&\ge 0,\\
			\left(\frac{1}{d}-\alpha-\beta\right)\left(\frac{1}{d}-\alpha(d-1)\right)-\alpha^2(d-1)&\ge 0.\label{eq:alpha,beta,d}
		\end{align}
		Then, it must hold that
		\begin{equation}
			\beta\le\frac{1-\alpha d^2}{d(1-\alpha d(d-1))},
		\end{equation}
		so we choose
		\begin{equation}
			\beta(\alpha):=\frac{1-\alpha d^2}{d(1-\alpha d(d-1))}\le\frac{1}{d}.
		\end{equation}
        $\beta(\alpha)$ is non-negative because from Eq.~\eqref{eq:alpha,beta,d}, $\alpha\le 1/d^2$ holds.
		In this case, the objective function becomes
		\begin{align}
			f(\alpha)&:=\alpha\Tr[\ketbra{\phi}{\phi}\ketbra{\phi}{\phi}]+\sum_{x\in\cX}\beta\Tr[\ketbra{\psi^*_x}{\psi^*_x}E^{x\top}_\rA]-\Tr\left[\frac{\one_\rR}{d}\right]\\
			&=       
			d^2\alpha+\frac{\sum_{x\in\cX}\|E^{x\top}_\rA\|}{d}\frac{1-\alpha d^2}{1-\alpha d(d-1)}-1.
		\end{align}
		Its derivative with respect to $\alpha$ is
		\begin{equation}
			f'(\alpha)=d^2-\frac{\sum_{x\in\cX}\|E^{x\top}_\rA\|}{(1-\alpha d(d-1))^2}.
		\end{equation}
		Then, for
		\begin{equation}
			\alpha^*=\frac{d-\sqrt{\sum_{x\in\cX}\|E^{x\top}_\rA\|}}{d^2(d-1)},
		\end{equation}
		we have $f'(\alpha^*)=0$.
		In this case, 
		\begin{equation}
			\beta(\alpha^*)=\frac{\sqrt{\sum_{x\in\cX}\|E^{x\top}_\rA\|}-1}{\sqrt{\sum_{x\in\cX}\|E^x_\rA\|}(d-1)}
		\end{equation}
		Since $1\le\sum_{x\in\cX}\|E^{x\top}_\rA\|\le d$, $\alpha^*,\beta(\alpha^*)\ge 0$.
		Also,
		\begin{equation}
			\frac{1}{d}-\alpha^*-\beta(\alpha^*)=\frac{\left(d-\sqrt{\sum_{x\in\cX}\|E^{x\top}_\rA\|}\right)^2}{d^2(d-1)\sqrt{\sum_{x\in\cX}\|E^{x\top}_\rA\|}}\ge 0
		\end{equation}
		and
		\begin{equation}
			\frac{1}{d}-\alpha^*(d-1)=\frac{\sqrt{\sum_{x\in\cX}\|E^{x\top}_\rA\|}}{d^2}\ge 0.
		\end{equation}
		Thus, $\alpha^*$ and $\beta(\alpha^*)$ lead to a dual feasible solution.
		Now we get
		\begin{equation}
			\begin{split}
				f(\alpha^*)=\frac{(\sqrt{\sum_{x\in\cX}\|E^{x\top}_\rA\|}-1)^2}{d-1}
			\end{split}
		\end{equation}
		Since the dual value of a feasible solution lower bounds the optimal value, we have
		\begin{equation}
			R(\rmid_\rA,\E^\cX_\rA)\ge \frac{(\sqrt{\sum_{x\in\cX}\|E^{x\top}_\rA\|}-1)^2}{d-1}
		\end{equation}
		Notice that the robustness of measurement is analytically calculated as follows~\cite{skrzypczyk2019robustness}:
		\begin{equation}
			I(\E^\cX_\rA)=\sum_{x\in\cX}\|E^{x}_\rA\|-1=\sum_{x\in\cX}\|E^{x\top}_\rA\|-1.
		\end{equation}
		Therefore, we get the desired inequality
		\begin{equation}
			R(\rmid_\rA,\E^\cX_\rA)\ge \frac{\left(\sqrt{I(\E^\cX_\rA)+1}-1\right)^2}{d-1}.
		\end{equation}
	\end{proof}
    
	\subsection{Comparison with Heinosaari--Miyadera (HM) bound}
	When $d:=\dim\cH_\rA\le 6$, the disturbance bound in Theorem~\ref{theorem:i-d-e} dominates the HM bound~\cite{heinosaari2013qualitative}, that is,
	\begin{equation}
		\label{eq:vsHM}
		\frac{2\Big(\sqrt{I(\E^\cX_\rA)+1}-1\Big)^{2}}{d-1}\ge\frac{1}{16}\sup_{x\in\cX}(\|E_\rA^x\|+\|\one_\rA-E_\rA^x\|-1)^{2}\quad(\forall \E^\cX_\rA\in\sM(\rA)).
	\end{equation}
	To show this, we use the expression $I(\E^\cX_\rA)=\sum_{x\in\cX}\|E^{x}_\rA\|-1$ ~\cite{skrzypczyk2019robustness}.
	From the triangle inequality, we have
	\begin{equation}
		\sum_{x\in\cX}\|E^x_\rA\|\ge \sup_{x\in\cX}\big(|E^x_\rA\|+\left\|\one_\rA-E^x_\rA\right\|\big).
	\end{equation}
	Setting $s:=\sqrt{\sup_{x\in\cX}(\|E_\rA^x\|+\|\one_\rA-E_\rA^x\|)}$ ($1\le s\le \sqrt{2}$), we can evaluate our lower bound as
	\begin{equation}
		\frac{2\left(\sqrt{I(\E^\cX_\rA)+1}-1\right)^2}{d-1}\ge\frac{2\left(s-1\right)^2}{d-1}.       
	\end{equation}
	On the other hand, the HM bound is expressed as
	\begin{equation}
		\frac{1}{16}\sup_{x\in\cX}(\|E_\rA^x\|+\|\one_\rA-E_\rA^x\|-1)^{2}=\frac{1}{16}(s^2-1)^2
	\end{equation}
	For $d=2,\dots,6$, we can see that
	\begin{equation}
	\frac{2\left(s-1\right)^2}{d-1}\ge\frac{1}{16}(s^2-1)^2
	\end{equation}
	holds for $1\le s\le \sqrt{2}$, which implies Eq.~\eqref{eq:vsHM}.    
        
	\subsection{Equivalence between incompatibility and informativeness in Example~\ref{example:unbiased}}
	We can show that Example~\ref{example:unbiased}, the unbiased binary Pauli-Z POVM, achieves the equality of Eq.~\eqref{eq:RvsRoI}.
	The RoI $R(\rmid_\rA,\mathsf{Z}(\eta)_\rA)$ is computed by the following SDP from Eq.~\eqref{sdp:channel-POVM_robustness}.
	\begin{equation}\label{eq:robustness_channel_povm}
    \begin{split}
		\underset{\tilde{C}_{\rA\rR}^{j}}{\minimize} \quad &r\ge 0\\
		\subto \quad &\sum_{j=\pm 1}\Tr_\rA\tilde{C}_{\rA\rR}^{j}=(1+r)\one_\rR,\\
		&\sum_{j=\pm 1}\tilde{C}_{\rA\rR}^{j}-\ketbra{\phi}{\phi}_{\rA\rR}\succeq \zero_{\rA\rR},\\
		&\Tr_{\rA}\tilde{C}_{\rA\rR}^{j}-Z(\eta)^{j\top}_\rA\succeq \zero_{\rA}\quad(j=\pm1)\\
		&\tilde{C}_{\rA\rR}^{j}\succeq\zero_{\rA\rR}\quad(j=\pm1).
        \end{split}
	\end{equation}
	
	Now we show $R(\rmid_\rA,\mathsf{Z}(\eta)_\rA)= (\sqrt{1+\eta}-1)^2$.
	One of the feasible solutions of Eq.~\eqref{eq:robustness_channel_povm} is
	\begin{align}
		\tilde{C}_{\rA\rR}^{+1}&= 8(\alpha\ket{00}_{\rA\rR}+\beta\ket{11}_{\rA\rR})(\alpha\bra{00}_{\rA\rR}+\beta\bra{11}_{\rA\rR})=8
		\begin{pmatrix}
			\alpha^2 &0&0& \alpha\beta\\
			0&0&0&0\\
			0&0&0&0\\
			\alpha\beta&0&0&\beta^2
		\end{pmatrix},\\
		\tilde{C}_{\rA\rR}^{-1}&=8(\beta\ket{00}_{\rA\rR}+\alpha\ket{11}_{\rA\rR})(\beta\bra{00}_{\rA\rR}+\alpha\bra{11}_{\rA\rR})=8
		\begin{pmatrix}
			\beta^2 &0&0& \alpha\beta\\
			0&0&0&0\\
			0&0&0&0\\
			\alpha\beta&0&0&\alpha^2
		\end{pmatrix},
	\end{align}
	where $\alpha=\frac{\sqrt{1+\eta}}{4}$ and $\beta=\frac{2-\sqrt{1+\eta}}{4}$.
	We will later use the following formulas:
	\begin{align}
		8(\alpha^2+\beta^2)&=3+\eta-2\sqrt{1+\eta}=1+(\sqrt{1+\eta}-1)^2,\\
		16\alpha\beta&=2\sqrt{1+\eta}-(1+\eta)=1-(\sqrt{1+\eta}-1)^2.
	\end{align}
	Let us confirm that the choices $\{\tilde{C}_{\rA\rR}^j\}_{j=\pm1}$ are feasible.
	First, they are clearly positive semidefinite.
	We also have
	\begin{align}
		\sum_{j=\pm1}\tilde{C}^j_{\rA\rR}-\ketbra{\phi}{\phi}_{\rA\rR}&=(\sqrt{1+\eta}-1)^2(\ket{00}_{\rA\rR}-\ket{11}_{\rA\rR})(\bra{00}_{\rA\rR}-\bra{11}_{\rA\rR})\succeq\zero_{\rA\rR},\\
		\Tr_{\rA}\tilde{C}_{\rA\rR}^{1}-Z(\eta)^{1\top}_\rA&=
		\begin{pmatrix}
			0 & 0\\
			0 & (\sqrt{1+\eta}-1)^2
		\end{pmatrix}
		\succeq \zero_{\rA},\\
		\Tr_{\rA}\tilde{C}_{\rA\rR}^{-1}-Z(\eta)^{-1\top}_\rA&=\begin{pmatrix}
			(\sqrt{1+\eta}-1)^2 & 0\\
			0 & 0
		\end{pmatrix}\succeq \zero_{\rA},
	\end{align}
	and
	\begin{equation}
    \sum_{j=\pm1}\Tr_\rA\tilde{C}_{\rA\rR}^{j}=8(\alpha^2+\beta^2)\one_\rR=(3+\eta-2\sqrt{1+\eta})\one_\rR=\{1+(\sqrt{1+\eta}-1)^2\}\one_\rR.
	\end{equation}
	Thus they are feasible solutions, and from the first constraint of Eq.~\eqref{eq:robustness_channel_povm}, for this solution, $r$ can be $(\sqrt{1+\eta}-1)^2$ that attains the equality of Eq.~\eqref{eq:RvsRoI}.

    \subsection{Another example}
    \begin{example}
		We consider the following POVM \cite{saini2025completeness}:
		\begin{equation}\label{eq:povm_example}
        \begin{split}
            &\left\{\frac{1-p}{3}\ketbra{+x}{+x}+\frac{p\one_\rA}{6},\frac{1-p}{3}\ketbra{-x}{-x}+\frac{p\one_\rA}{6},\right.\\
			&\left.\quad\frac{1-p}{3}\ketbra{+y}{+y}+\frac{p\one_\rA}{6},\frac{1-p}{3}\ketbra{-y}{-y}+\frac{p\one_\rA}{6},\right.\\
			&\left.\quad\frac{1-p}{3}\ketbra{+z}{+z}+\frac{p\one_\rA}{6},\frac{1-p}{3}\ketbra{-z}{-z}+\frac{p\one_\rA}{6}\right\}\quad (p\in[0,1]).
        \end{split}
		\end{equation}
		One has
        \begin{align}
            \left\|\frac{1-p}{3}\ketbra{+x}{+x}+\frac{p\one_\rA}{6}\right\|=\frac{2-p}{6},\quad
            \left\|\one_\rA-\left(\frac{1-p}{3}\ketbra{+x}{+x}+\frac{p\one_\rA}{6}\right)\right\|=\frac{6-p}{6}.
        \end{align}
		Therefore, the robustness of measurement is $1-p$ and our lower bound is $2(\sqrt{2-p}-1)^2$, while
		the HM bound is $\frac{(1-p)^{2}}{144}$.
		These two bounds are compared in Fig.~\ref{fig:povm_example}. 
	\end{example}
	\begin{figure}[h]
		\centering
		\includegraphics[width=0.6\linewidth]{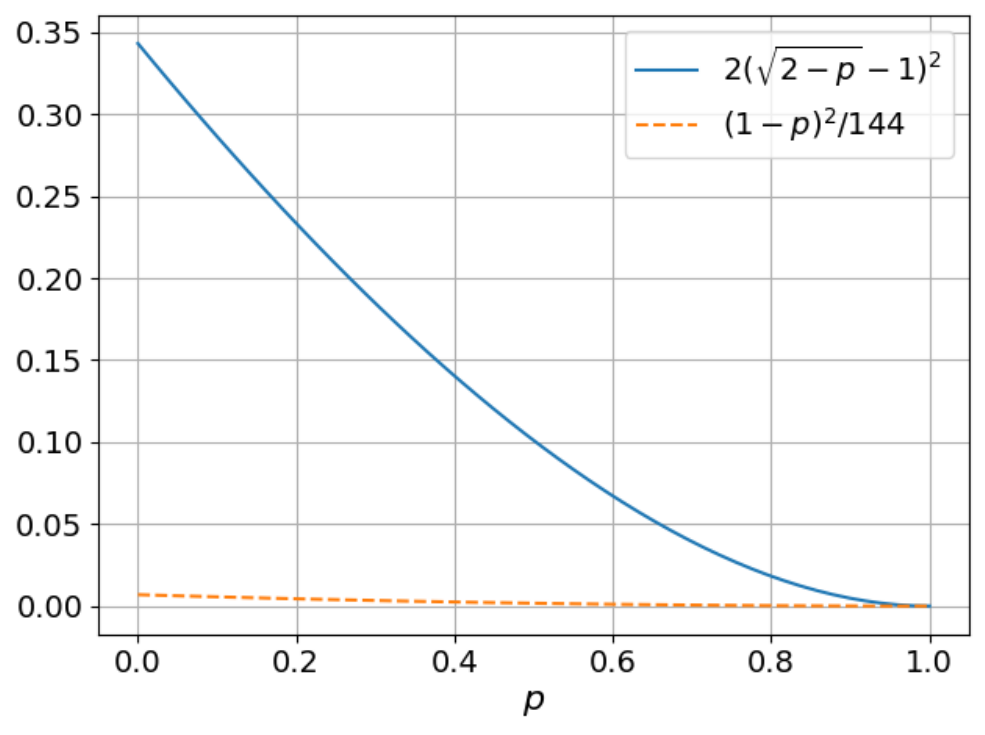}
		\caption{Comparison with our robustness-based bound $2(\sqrt{2-p}-1)^2$ and Heinosaari--Miyadera's bound $(1-p)^2/144$ for the POVM Eq.~\eqref{eq:povm_example}.}
		\label{fig:povm_example}
	\end{figure}

    \section{Weight}\label{appendix:weight}
	For a channel pair $(\Lambda_{\rA\to\rB_1},\Xi_{\rA\to\rB_2})$, its weight of incompatibility (WoI) $W(\Lambda_{\rA\to\rB_1},\Xi_{\rA\to\rB_2})$ is defined as follows:
	\begin{equation}
		\begin{split}
			W(\Lambda_{\rA\to\rB_1},\Xi_{\rA\to\rB_2}):=&\underset{\Theta_{\rA\to\rB_1\rB_2},\cN^{(1)}_{\rA\to\rB_1},\cN^{(2)}_{\rA\to\rB_2}}{\minimize}\quad w\ge 0\\
			\subto \quad &(1-w)\Tr_{\rB_2}\circ\Theta_{\rA\to\rB_1\rB_2}+w\cN^{(1)}_{\rA\to\rB_1}=\Lambda_{\rA\to\rB_1},\\
			&(1-w)\Tr_{\rB_1}\circ\Theta_{\rA\to\rB_1\rB_2}+w\cN^{(2)}_{\rA\to\rB_2}=\Xi_{\rA\to\rB_2},\\
			&\cN^{(1)}_{\rA\to\rB_1}\in\sC(\rA\to\rB_1),\quad \cN^{(2)}_{\rA\to\rB_2}\in\sC(\rA\to\rB_2),\\
			&\Theta_{\rA\to\rB_1\rB_2}\in\sC(\rA\to\rB_1\rB_2).
		\end{split}
	\end{equation}
	In the formalism of Choi operators, it is
	\begin{equation}
		\begin{split}
			W(\Lambda_{\rA\to\rB_1},\Xi_{\rA\to\rB_2})=&\underset{C_{\rB_1\rB_2\rR},C_{\rB_1\rR}(\cN^{(1)}_{\rA\to\rB_1}),C_{\rB_2\rR}(\cN^{(2)}_{\rA\to\rB_2})}{\minimize}\quad  w\ge 0\\
			\subto\quad &(1-w)\Tr_{\rB_2}[C_{\rB_1\rB_2\rR}]+wC_{\rB_1\rR}(\cN^{(1)}_{\rA\to\rB_1})=C_{\rB_1\rR}(\Lambda_{\rA\to\rB_1})\\
			&(1-w)\Tr_{\rB_1}[C_{\rB_1\rB_2\rR}]+wC_{\rB_2\rR}(\cN^{(2)}_{\rA\to\rB_2})=C_{\rB_2\rR}(\Xi_{\rA\to\rB_2})\\
			&\Tr_{\rB_1}[C_{\rB_1\rR}(\cN_{\rA\to\rB_1}^{(1)})]=\one_{\rR},\quad C_{\rB_1\rR}(\cN_{\rA\to\rB_1}^{(1)})\succeq\zero_{\rB_1\rR}\\
			&\Tr_{\rB_2}[C_{\rB_2\rR}(\cN_{\rA\to\rB_2}^{(2)})]=\one_{\rR},\quad C_{\rB_2\rR}(\cN_{\rA\to\rB_2}^{(2)})\succeq\zero_{\rB_2\rR}\\
			&C_{\rB_1\rB_2\rR}\succeq\zero_{\rB_1\rB_2\rR}.
		\end{split}
	\end{equation}
	By introducing 
	\begin{equation}
		\tilde{C}_{\rB_1\rB_2\rR}:=(1-w)C_{\rB_1\rB_2\rR},
	\end{equation}
	and assuming $w>0$, we obtain
	\begin{equation}
		\begin{split}
			W(\Lambda_{\rA\to\rB_1},\Xi_{\rA\to\rB_2})=&\underset{\tilde{C}_{\rB_1\rB_2\rR},C_{\rB_1\rR}(\cN^{(1)}_{\rA\to\rB_1}),C_{\rB_2\rR}(\cN^{(2)}_{\rA\to\rB_2})}{\minimize}\quad  w>0\\
			\subto\quad &C_{\rB_1\rR}(\cN^{(1)}_{\rA\to\rB_1})=\frac{1}{w}\left(C_{\rB_1\rR}(\Lambda_{\rA\to\rB_1})-\Tr_{\rB_2}[\tilde{C}_{\rB_1\rB_2\rR}]\right)\\
			&C_{\rB_2\rR}(\cN^{(2)}_{\rA\to\rB_2})=\frac{1}{w}\left(C_{\rB_2\rR}(\Xi_{\rA\to\rB_2})-\Tr_{\rB_1}[\tilde{C}_{\rB_1\rB_2\rR}]\right)\\
			&\Tr_{\rB_1}[C_{\rB_1\rR}(\cN_{\rA\to\rB_1}^{(1)})]=\one_{\rR},\quad C_{\rB_1\rR}(\cN_{\rA\to\rB_1}^{(1)})\succeq\zero_{\rB_1\rR}\\
			&\Tr_{\rB_2}[C_{\rB_2\rR}(\cN_{\rA\to\rB_2}^{(2)})]=\one_{\rR},\quad C_{\rB_2\rR}(\cN_{\rA\to\rB_2}^{(2)})\succeq\zero_{\rB_2\rR}\\
			&\tilde{C}_{\rB_1\rB_2\rR}\succeq\zero_{\rB_1\rB_2\rR}.
		\end{split}
	\end{equation}
	Thus computing the WoI becomes the following SDP:
	\begin{equation}
		\begin{split}
			W(\Lambda_{\rA\to\rB_1},\Xi_{\rA\to\rB_2})=\underset{\tilde{C}_{\rB_1\rB_2\rR}}{\minimize}\quad & w\ge0\\
			\subto\quad             &\Tr_{\rB_1\rB_2}[\tilde{C}_{\rB_1\rB_2\rR}]=(1-w)\one_{\rR}\\ &C_{\rB_1\rR}(\Lambda_{\rA\to\rB_1})-\Tr_{\rB_2}[\tilde{C}_{\rB_1\rB_2\rR}]\succeq\zero_{\rB_1\rR}\\
			&C_{\rB_2\rR}(\Xi_{\rA\to\rB_2})-\Tr_{\rB_1}[\tilde{C}_{\rB_1\rB_2\rR}]\succeq\zero_{\rB_2\rR}\\
			&\tilde{C}_{\rB_1\rB_2\rR}\succeq\zero_{\rB_1\rB_2\rR}.
		\end{split}
	\end{equation}
	The Lagrangian is
	\begin{equation}
		\begin{split}
			L&:=w-\Tr[V_\rR(\Tr_{\rB_1\rB_2}[\tilde{C}_{\rB_1\rB_2\rR}]-(1-w)\one_\rR)] -\Tr[X_{\rB_1\rR}(C_{\rB_1\rR}(\Lambda_{\rA\to\rB_1})-\Tr_{\rB_2}[\tilde{C}_{\rB_1\rB_2\rR}])] \\
			&\quad-\Tr[Y_{\rB_2\rR}(C_{\rB_2\rR}(\Xi_{\rA\to\rB_2})-\Tr_{\rB_1}[\tilde{C}_{\rB_1\rB_2\rR}])] -\Tr[Z_{\rB_1\rB_2\rR}\tilde{C}_{\rB_1\rB_2\rR}]\\
			&=\Tr[V_\rR]-\left( \Tr[X_{\rB_1\rR}C_{\rB_1\rR}(\Lambda_{\rA\to\rB_1})]+\Tr[Y_{\rB_2\rR}C_{\rB_2\rR}(\Xi_{\rA\to\rB_2})]\right)\\
			&\quad -w(\Tr[V_\rR]-1)-\Tr[(-X_{\rB_1\rR}\otimes\one_{\rB_2}-\one_{\rB_1}\otimes Y_{\rB_2\rR}+\one_{\rB_1\rB_2}\otimes V_\rR+Z_{\rB_1\rB_2\rR})\tilde{C}_{\rB_1\rB_2\rR}],
		\end{split}
	\end{equation}
	where $V_\rR\in\cL_\rH(\cH_\rR)$, $X_{\rB_1\rR}\succeq\zero_{\rB_1\rR}$, $Y_{\rB_2\rR}\succeq\zero_{\rB_2\rR}$, $Z_{\rB_1\rB_2\rR}\succeq\zero_{\rB_1\rB_2\rR}$, and $w\ge0$. 
	Therefore, the dual SDP is as follows:
	\begin{equation}
		\begin{split}
			\underset{X_{\rB_1\rR},Y_{\rB_2\rR},V_\rR}{\text{Maximize}} \quad & \Tr[V_\rR]-\Big( \Tr[X_{\rB_1\rR}C_{\rB_1\rR}(\Lambda_{\rA\to\rB_1})]+\Tr[Y_{\rB_2\rR}C_{\rB_2\rR}(\Xi_{\rA\to\rB_2})]
			\Big)\\
			\text{subject to} \quad 
			& \Tr[V_{\mathrm{R}}]\le 1, V_\rR\in\cL_\cH(\cH_\rR)\\
			& X_{\rB_1\rR}\otimes\one_{\rB_2}+\one_{\rB_1}\otimes Y_{\rB_2\rR}\succeq\one_{\rB_1\rB_2}\otimes V_{\rR}\\
			& X_{\rB_1\rR}\succeq\zero_{\rB_1\rR},\quad Y_{\rB_2\rR}\succeq\zero_{\rB_2\rR}.
		\end{split}
	\end{equation}
	The same argument as RoI makes the above problem as follows:
	\begin{equation}
		\begin{split}
			\underset{X_{\rB_1\rR},Y_{\rB_2\rR},V_\rR}{\text{Maximize}} \quad & 1-\Big( \Tr[X_{\rB_1\rR}C_{\rB_1\rR}(\Lambda_{\rA\to\rB_1})]+\Tr[Y_{\rB_2\rR}C_{\rB_2\rR}(\Xi_{\rA\to\rB_2})]
			\Big)\\
			\text{subject to} \quad 
			& \Tr[V_{\mathrm{R}}]=1, V_\rR\in\cL_\cH(\cH_\rR)\\
			& X_{\rB_1\rR}\otimes\one_{\rB_2}+\one_{\rB_1}\otimes Y_{\rB_2\rR}\succeq\one_{\rB_1\rB_2}\otimes V_{\rR}\\
			& X_{\rB_1\rR}\succeq\zero_{\rB_1\rR},\quad Y_{\rB_2\rR}\succeq\zero_{\rB_2\rR}.
		\end{split}
	\end{equation}
	
	Unlike RoI, the inequality in the second constraint is reversed for WoI.
	Therefore, developing a direct connection to the diamond norm distance is not straightforward.
\end{document}